\documentclass[onecolumn,11pt]{IEEEtran}  % IEEE arxiv Style

\usepackage{amsmath}
\usepackage{amssymb}                                           
\usepackage{MnSymbol}
\usepackage{amsxtra,amsthm}
\usepackage{bbold}
\usepackage{mathdots}
\usepackage{mathtools}
\usepackage{pdfcolmk}
\usepackage{graphicx}                                        % Eps Grafiken für PostScript Drucker
\usepackage{algorithm}
\usepackage{algorithmic} 
\usepackage{mathrsfs}
\usepackage{todonotes}

\usepackage[utf8]{inputenc}                                         % -> damit man sie auch gleich im Code eingeben kann
\usepackage[T1]{fontenc}
\usepackage{verbatim}
\usepackage{srcltx}
\usepackage{csquotes}
\usepackage[english]{babel}

\usepackage[shortlabels]{enumitem}
\usepackage[all]{xy}
\usepackage[style=ieee,    % https://www.ctan.org/tex-archive/macros/latex/contrib/biblatex-contrib/biblatex-ieee
            firstinits=true, % vornamen = erste buchstabe
            url=false,
            backend=bibtex,
            doi=false,
            maxbibnames=99,
            isbn=false,
            natbib=true,
            hyperref=false,
            bibencoding=utf8]{biblatex}
\AtEveryBibitem{\clearfield{note}}   % damit mein komentare nicht geprintet werden.
\DefineBibliographyStrings{english}{%
  references = {},% replace "references" with "bibliography"  for `book`/`report`
}

\bibliography{abbrv,jabref_philipp_utf2}

\newcommand{\mgeq}{\succeq}

\newcommand{\Lone}{{{L}_1}}
\newcommand{\Ltwo}{{{L}_2}}

\let\forallalt\forall
\renewcommand{\forall}{\;\forallalt\;}

\let\refalt\ref
\renewcommand{\ref}[1]{(\refalt{#1})}

\newcommand{\alp}{\ensuremath{\alpha}}

\newcommand{\lam}{\ensuremath{\lambda}}

\newcommand{\vlam}{\ensuremath{\boldsymbol{ \lam}}}

\newcommand{\va}{{\ensuremath{\mathbf a}}}
\newcommand{\vb}{{\ensuremath{\mathbf b}}}
\newcommand{\vc}{{\ensuremath{\mathbf c}}}
\newcommand{\vd}{{\ensuremath{\mathbf d}}}
\newcommand{\vh}{{\ensuremath{\mathbf h}}}                         % Vektorwertige Funktion
\newcommand{\vn}{{\ensuremath{\mathbf n}}}                         % Vektorwertige Funktion

\newcommand{\vx}{{\ensuremath{\mathbf x}}}

\newcommand{\vxone}{{\ensuremath{\mathbf x}_1}}
\newcommand{\vxtwo}{{\ensuremath{\mathbf x}_2}}

\newcommand{\vtxone}{{\ensuremath{\tilde{\mathbf x}_1}}}
\newcommand{\vtxtwo}{{\ensuremath{\tilde{\mathbf x}_2}}}

\newcommand{\vy}{{\ensuremath{\mathbf y}}}

\newcommand{\vz}{{\ensuremath{\mathbf z}}}

\newcommand{\Rvx}{{\ensuremath{\mathbf x}^{-}}}
\newcommand{\Rvxj}{{\ensuremath{\mathbf x}^{-}_j}}

\newcommand{\Rvxone}{\ensuremath{{\mathbf x}_1^{-}}}
\newcommand{\Rvxtwo}{{\ensuremath{\mathbf x}_2^{-}}}

\newcommand{\vta}{\ensuremath{\tilde{\mathbf a}}}

\newcommand{\vtx}{\ensuremath{\tilde{\mathbf x}}}

\newcommand{\uA}{{\ensuremath{\mathrm A}}}
\newcommand{\uB}{{\ensuremath{\mathrm B}}}
\newcommand{\uC}{{\ensuremath{\mathrm C}}}

\newcommand{\uI}{{\ensuremath{\mathrm I}}}
\newcommand{\uX}{{\ensuremath{\mathrm X}}}
\newcommand{\uXone}{{\ensuremath{\mathrm X}_1}}
\newcommand{\uXtwo}{{\ensuremath{\mathrm X}_2}}
\newcommand{\tuX}{{\ensuremath{\tilde{\mathrm X}}}}
\newcommand{\uY}{{\ensuremath{\mathrm Y}}}

\newcommand{\uR}{{\ensuremath{\mathrm R}}}
\newcommand{\uH}{{\ensuremath{\mathrm H}}}
\newcommand{\uS}{{\ensuremath{\mathrm S}}}
\newcommand{\uG}{{\ensuremath{\mathrm G}}}
\newcommand{\tuG}{{\ensuremath{\tilde{\mathrm G}}}}
\newcommand{\tuI}{{\ensuremath{\tilde{\mathrm I}}}}
\newcommand{\tuR}{{\ensuremath{\tilde{\mathrm R}}}}
\newcommand{\tuS}{{\ensuremath{\tilde{\mathrm S}}}}
\newcommand{\uN}{{\ensuremath{\mathrm N}}}
\newcommand{\uT}{{\ensuremath{\mathrm T}}}

\newcommand{\uP}{{\ensuremath{\mathrm P}}}
\newcommand{\uQ}{{\ensuremath{\mathrm Q}}}

\newcommand{\vA}{{\ensuremath{\mathbf A}}}

\newcommand{\vC}{{\ensuremath{\mathbf C}}}

\newcommand{\vH}{\ensuremath{\mathbf H }}                         % Vektorwertige Funktion

\newcommand{\vR}{\ensuremath{\mathbf R}}
\newcommand{\vS}{\ensuremath{\mathbf S }}                         % Vektorwertiges Ma{\ss}
\newcommand{\vT}{\ensuremath{\mathbf T}}

\newcommand{\vW}{{\ensuremath{\mathbf W}}}
\newcommand{\vX}{{\ensuremath{\mathbf X}}}

\newcommand{\zero}{{\ensuremath{\mathbf 0}}}

\newcommand{\Torus}{\ensuremath{\mathbb{T}} }

\newcommand{\Alin}{{\ensuremath{\mathcal{A}}}}
\newcommand{\tAlin}{{\ensuremath{\tilde{\mathcal{A}}}}}

\newcommand{\C}{{\ensuremath{\mathbb C}}}

\newcommand{\R}{{\ensuremath{\mathbb R}}}
\newcommand{\K}{{\ensuremath{\mathbb F}}} % Exclusive jetzt für Field F

\newcommand{\Z}{{\ensuremath{\mathbb Z}}}

\newcommand{\Zform}{{\ensuremath{\mathcal Z}}}

\newcommand{\id}{{\ensuremath{\mathbb 1}}}

\newcommand{\RA}{\ensuremath{\Rightarrow} }

\newcommand{\LRA}{\ensuremath{\Leftrightarrow} }

\newcommand{\Pro}{\prod}
\newcommand{\Proj}{{\mathbf \Pi}}
\newcommand{\ProjNL}{\ensuremath{{\mathbf \Pi}_{N,L}}}

\newcommand{\ProjNLj}{\ensuremath{{\mathbf \Pi}_{N,L_j}}}

\newcommand{\skprod}[2]{\ensuremath{ \left\langle #1,#2 \right\rangle }}

\definecolor{gray}{rgb}{0.3,0.3,0.3}
{\color{black}}                      % change to gray if you want see any changes
{\color{black}}

{\color{black}}                      % change to gray if you want see any changes
{\color{black}}
\newcommand{\thmref}[1]{Theorem~\ref{#1}}     % Theorem
\newcommand{\lemref}[1]{Lemma~\ref{#1}}       % Lemma
\newcommand{\enuref}[1]{\ref{#1}}             % Enumerations (i) etc..

\newcommand{\appref}[1]{Appendix~\ref{#1}}

\newcommand{\figref}[1]{Figure~\ref{#1}}

\newcommand{\noi}{\noindent}

\ifx\definition\undefined
\newtheorem{definition}{Definition}         %%%%%%%%%%%%%%%%%%%%%%%%%%%%%%
\fi
\ifx \@definition \@empty
\newtheorem{definition}[theorem]{Definition}   % the section where they     %
\fi
\ifx\corrolary\undefined
\newtheorem{corrolary}{Corrolary} %%%%%%%%%%%%%%%%%%%%%%%%%%%%%% 
\fi
\ifx\conjecture\undefined
\newtheorem{conjecture}{Conjecture}         %%%%%%%%%%%%%%%%%%%%%%%%%%%%%%
\fi
\ifx\theorem\undefined
\newtheorem{theorem}{Theorem}         %%%%%%%%%%%%%%%%%%%%%%%%%%%%%%
\fi
\ifx\lemma\undefined
\newtheorem{lemma}{Lemma}
\fi
\ifx\question\undefined
\newtheorem{question}{\color{red}{Question}}         %%%%%%%%%%%%%%%%%%%%%%%%%%%%%%
\fi

\ifx\remark\undefined
\newenvironment{remark}{\par\vspace{1.5ex}\noindent{\em Remark\/}.}{\par\vspace{1.5ex}}
\fi

\DeclareMathOperator{\spann}{span}

\DeclareMathOperator{\range}{range}

\DeclareMathOperator{\find}{find}

\newcommand\rank{\operatorname{rank}}

\newcommand\tr{\operatorname{tr}} 

\newcommand\kernel{\operatorname{kern}}

\newcommand{\argmin}[1]{\underset{#1}{\operatorname{argmin}}}

\newcommand{\Norm}[1]{\ensuremath{ \left\|#1\right\| }}

\newcommand{\Expect}[1]{{\ensuremath{\mathbb E}[#1]}}

\newcommand{\cc}[1]{{\ensuremath{\overline{#1}}}} % complex conjugation

\makeatletter
  \newcommand{\set}[2]{\ensuremath{%
  \setbox0=\hbox{\ensuremath{#2}}
  \dimen@\ht0
  \advance\dimen@ by \dp0
  \left\{\left.#1\rule[-\dp0]{0pt}{\dimen@}\;\right|\;#2\right\} }}
\makeatother

\newcommand{\Rmap}{\ensuremath{\mathfrak R}}

\renewcommand{\argmin}{\operatornamewithlimits{argmin}}

\newcommand{\namen}[1]{{\textsc{#1}}}           % englische w{\"o}rter als kursiv schreiben
\ifx \@paragraph \@empty
\makeatletter
\renewcommand\paragraph{\@startsection
{paragraph}{4}{\z@}{-3.5ex plus-1ex minus-.2ex}%
{1.3ex plus.2ex}{\normalfont\itshape}}
\makeatother
\fi

\renewcommand{\Re}{\ensuremath{\operatorname{Re}}}
\renewcommand{\Im}{\ensuremath{\operatorname{Im}}}
\newcommand{\textIm}{\ensuremath{\Im}} % imaginary
\newcommand{\textRe}{\ensuremath{\Re}} % real 
  \newcommand{\seefor}[1]{\!}
  \newcommand{\seeintern}[1]{\!}

\renewcommand{\alp}{\phi}
\usepackage{tikz,pgfplots}
\usepackage{filecontents}
\usetikzlibrary{positioning,decorations.pathreplacing}

  \newcommand{\tikzmark}[1]{\tikz[overlay,remember picture,baseline=(#1.base)] \node (#1) {\strut};}

  \newcommand{\tikzmarkright}[1]{\tikz[overlay,remember picture,baseline=(#1.base)] \node [right=4pt] (#1)   {\strut};}

  \newcommand{\tzm}[1]{\tikzmark{#1}}
  
  \newcommand{\tzmr}[1]{\tikzmarkright{#1}}

\begin{document}
\title{Blind Deconvolution with Additional Autocorrelations via Convex Programs}

\author{
\IEEEauthorblockN{Philipp Walk\IEEEauthorrefmark{1}, Peter Jung\IEEEauthorrefmark{2}, Götz E. Pfander\IEEEauthorrefmark{3} 
     and Babak Hassibi\IEEEauthorrefmark{1}\\}
\IEEEauthorblockA{\IEEEauthorrefmark{1}Department of Electrical Engineering, Caltech, Pasadena, CA 91125\\
Email: \{pwalk,hassibi\}@caltech.edu \\}
\IEEEauthorblockA{\IEEEauthorrefmark{2}Communications \& Information Theory, Technical University Berlin, 10587 Berlin\\
Email: peter.jung@tu-berlin.de\\}
\IEEEauthorblockA{\IEEEauthorrefmark{3}Philipps-University Marburg, Mathematics \& Computer Science\\
Email: pfander@mathematik.uni-marburg.de}
}

\maketitle

\if0 % abstract for the webfront end (will be in the proceedings
     % content pages (50-100 words)
\fi
\begin{abstract}
  In this work we characterize all ambiguities of the linear (aperiodic) one-dimensional convolution on two fixed
  finite-dimensional complex vector spaces. It will be shown that the convolution ambiguities can be mapped one-to-one
  to factorization ambiguities in the $z-$domain, which are generated by swapping the zeros of the input signals. We use
  this polynomial description to show a deterministic version of a recently introduced masked Fourier phase retrieval
  design.  In the noise-free case a (convex) semidefinite program can be used to recover exactly the input signals if
  they share no common factors (zeros).  Then, we reformulate the problem as deterministic blind deconvolution with
  prior knowledge of the autocorrelations. Numerically simulations show that our approach is also robust against
  additive noise. 
\end{abstract}

%%%%%%%%%%%%%%%%%%%%%%%%%%%%%%%%%%%%%%%%%%%%%%%%%%%%%%%%%%%%%%%%%%%%%%%%%%%%%%%%%%%%%%%%%%%%%%%%%%%%%%%%%%%%%%%%%%%%%%%%%%
\section{Introduction}
%%%%%%%%%%%%%%%%%%%%%%%%%%%%%%%%%%%%%%%%%%%%%%%%%%%%%%%%%%%%%%%%%%%%%%%%%%%%%%%%%%%%%%%%%%%%%%%%%%%%%%%%%%%%%%%%%%%%%%%%%%
% vielleicht hier noch ein bisschen mehr zu deconvolution problems

Blind deconvolution problems occur in many signal processing applications, as in digital communication over wire or
wireless channels. Here, the channel (system), usually assumed to be linear time invariant, has to be identified or
estimated at the receiver. Once, the channel can be probed sufficiently often and the channel parameter stay constant
over a longer period, pilot signals can be used for this purpose. However, in some cases one also has to estimate or
equalize the channel blindly.  Blind channel equalization and estimation methods were already developed in the $90$ties,
see for example in \cite{TXK91,DKAJ91,TXHK95} for the case where the receiver has statistical channel knowledge, for
example second order or higher moments.  If no statistical knowledge of the data and the channel is available, for
example, for fast fading channels, one can still ask under which conditions on the data and the channel a blind channel
identification is possible. Necessary and sufficient conditions in a multi-channel setup where first derived in
\cite{XLTK95,GN95} and continuously further developed, see e.g.  \cite{A-MQH97} for a nice summary. All these techniques
are of iterative nature which are therefor difficult to analyze. Most of the algorithms often suffer from instabilities
in the presence of noise and overall the performance is inadequate for many applications.  To overcome these
difficulties, we will propose in this work a convex program for simultaneous reconstruction of the channel and data
signal. We show that this program is always successful in the noiseless setting and we numerically demonstrate its
stability under noise. The blind reconstruction can hereby be re-casted as a phase retrieval problem if we have
additional knowledge of the autocorrelation of the data and the channel at the receiver, which was shown  by Jaganathan
and one of the authors in \cite{JH16}.  The uniqueness of the phase retrieval problem can then be shown by constructing
an explicit dual certificate in the noise free case by translating the ideas of \cite{JH16} to a purely deterministic
setting. We show that the convex program derived in \cite{JH16} holds indeed for every signal and channel of fixed
dimensions as long as the corresponding $z-$transforms have no common zeros, which is known to be a necessary condition
for blind deconvolution \cite{XLTK95}. Before we propose the new blind deconvolution setup we will define and analyse
all ambiguities of (linear) convolutions in finite dimensions.

%%%%%%%%%%%%%%%%%%%%%%%%%%%%%%%%%%%%%%%%%%%%%%%%%%%%%%%%%%%%%%%%%%%%%%%%%%%%%%%%%%%%%%%%%%%%%%%%%%%%%%%%%%%%%%%%%%%%%%%%%%
\section{Ambiguities of Convolution}
%%%%%%%%%%%%%%%%%%%%%%%%%%%%%%%%%%%%%%%%%%%%%%%%%%%%%%%%%%%%%%%%%%%%%%%%%%%%%%%%%%%%%%%%%%%%%%%%%%%%%%%%%%%%%%%%%%%%%%%%%%

The convolution defines a product and it is therefore obvious that this comes with factorization ambiguities. But, so
far, the authors couldn't find a mathematical rigorous and complete characterization and definition of all convolution
ambiguities in the literature.  Even in the case of autocorrelations, as investigated in phase retrieval problems, the
definition of ambiguities seems at least not consistent, see for example \cite{Hur89, BS79} or even a recent work
\cite{BP15}. To obtain well-posed blind deconvolution problems of finite dimensional vectors, we have to precisely
define all ambiguities of convolutions over the field $\C$ in the finite dimensions $\Lone$ respectively $\Ltwo$.  Only
if we  exclude all non-trivial ambiguities we obtain identifiability of the inputs $(\vxone,\vxtwo)\in\C^{\Lone}\times
\C^{\Ltwo}$ from their \emph{aperiodic or linear convolution product} $\vy\in\C^{L_1+L_2-1}$, given component-wise for
$k\in\{0,1,\dots, \Lone\!+\!\Ltwo\!-\!2\}=:[\Lone\!+\!\Ltwo\!-\!1]$ as 
\begin{align}
 y_k= (\vxone*\vxtwo)_k := \sum_{l=0}^{\min\{\Lone\!-\!1,k\}} x_{1,l} x_{2,k-l}.\label{eq:convtime}
\end{align}
A first analytic characterization of such identifiable classes, also for general bilinear maps, in the time domain
$\C^{\Lone}\times \C^{\Ltwo}$ was obtained in \cite{WJ12b,CM13a,CM15a}. 
However, before we define the convolution ambiguities, we will define first the scaling ambiguity in $\C^{\Lone}\times
\C^{\Ltwo}$ which is the intrinsic ambiguity of scalar multiplication $m\colon\C\times \C\to \C$ mapping any pair
$(a,b)$ to the product $m(a,b):=ab$.  Obviously, this becomes the only ambiguity if any bilinear map, as the
convolution, is defined for trivial dimensions $\Lone=\Ltwo=1$.
%
%which is nothing else as a (scaled) scalar multiplication $m$.
%
We have therefore the following definition.  
%
%%%%%%%%%%%%%%%%%%%%%%%%%%%%%%%%%%%%%%%%%%%%%%%%%%%%%%%%%%%%%%%%%%%%%%%%%%%%%%%%%%%%%%%%%%%%%%%%%%%%%%%%%%%%%%%%%%%%%%%%%%
\begin{definition}[Scaling Ambiguities]
  Let $L_1,L_2$ be positive integers. Then the scalar multiplication $m$ in $\C$ induces a \emph{scaling equivalence relation} on
  $\C^{\Lone}\times\C^{\Ltwo}$ defined by
  \begin{align}
    (\vxone, \vxtwo) \sim_{m} (\vtxone,\vtxtwo) \LRA \exists\lam\in\C\colon \vtxone=\lam\vxone,
    \vtxtwo=\lam^{-1}\vxtwo.\label{eq:linconv}
  \end{align}
  We call $[(\vxone,\vxtwo)]_m:= \set{(\vtxone,\vtxtwo)}{(\vxone,\vxtwo)\sim_m (\vtxone,\vtxtwo)}$ the \emph{scaling
  equivalence class} of $(\vxone,\vxtwo)$.
\end{definition}
%%%%%%%%%%%%%%%%%%%%%%%%%%%%%%%%%%%%%%%%%%%%%%%%%%%%%%%%%%%%%%%%%%%%%%%%%%%%%%%%%%%%%%%%%%%%%%%%%%%%%%%%%%%%%%%%%%%%%%%%%%
%
\begin{figure}[t]
  \centering
  \xymatrixcolsep{1pc}
  \xymatrixrowsep{0.4pc}
  $\xymatrix{ 
    \text{Root-Domain} & &z\text{-Domain} & & \! \! \! \text{Time-Domain}\\
    \C^{\Lone-1}\times\C^{\Ltwo-1} \ar[rr]& & 
    \ar[ll] \C_{\Lone}[z]\times \C_{\Ltwo}[z] \ar[dd]_{\cdot}
    \ar[rr]&&  \ar[ll] \C^{\Lone} \times \C^{\Ltwo} \ar[dd]^{*}  \\ \\
    \C^{\Lone+\Ltwo-2} \ar[uu]^{\Pi}\ar[rr]  &  & \ar[ll]_{\Rmap\quad\quad }\C_{\Lone +\Ltwo-2}[z]  \ar[rr]^{\quad\Zform} &&
   \ar[ll]  \C^{\Lone+\Ltwo-1}
     }$
     \caption{Zero/root representation of the convolution}\label{fig:zerorepconv}
\end{figure}
\begin{remark}
  The scaling ambiguity can be easily generalized over any field $\K$. 
\end{remark}

We identify $\vx\in\C^N$ with its \emph{one-sided} or \emph{unilateral $z-$transform} or \emph{transfer function}, given by
\begin{align}
  \uX(z)=(\Zform \vx)(z):=\sum_{k=0}^{N-1} x_k z^{-k} = \sum_{k=F}^D x_k z^{-k},\label{eq:zform}
\end{align}
where $D$ denotes the largest (degree of $\uX$) and $F$ the smallest non-zero coefficient index of $\vx$. The transfer
function in \eqref{eq:zform} is also called and \emph{FIR filter} or \emph{all-zero filter}, i.e., the only pole is
attained at $z=0$, and if the first coefficient is not vanishing all zeros are finite (lying in a circle of finite
radius), see \figref{fig:x1x2ambi} and \figref{fig:ambiguityautcorN}. Here, $\uX\in\C[z]$ defines a polynomial over
$z^{-1}$ and therefor we will not distinguish in the sequel between polynomial and unilateral $z-$transform. The set of
all finite degree polynomials $\C[z]$ defines with the polynomial multiplication $\cdot$ (algebraic convolution) 
\begin{align}
 \uY(z)= \uX_1(z)\cdot \uX_2(z):= \sum_{l=0}^{\Lone-1} x_{1,l} z^{-l}\cdot \sum_{l=0}^{\Ltwo-1} x_{2,l}
  z^{-l}=\sum_{k=0}^{\Lone+\Ltwo-2} \left(\sum_{l=0}^{\min\{\Lone\!-\!1,k\}} x_{1,l} x_{2,k-l}\right) z^{-k}\label{eq:zconv}
\end{align}
a ring, called the \emph{polynomial ring}.  Since $\C$ is an algebraically closed field we have, up to a unit $u\in\C$,
a \emph{unique factorization} of $\uX\in\C[z]$ of degree $D$ in primes $\uP_k(z):=z^{-1}-\zeta_k^{-1}$ (irreducible polynomials
of degree one), i.e., 
\begin{align}
  \uX(z) = x_{F} \Pro_{k=1}^{D} (z^{-1}-\zeta^{-1}_k),\label{eq:uniquefacX}
\end{align}
is determined by the $D$ \emph{zeros}  $\zeta_k$ of $\uX$ and the unit $x_F$.  Hence, for finite-length sequences
(vectors), the linear convolution \eqref{eq:linconv} can be represented with the $z-$transform $\Zform$ one-to-one in
the $z-$domain as the polynomial multiplication \eqref{eq:zconv}, see for example the classical text books \cite{Lan02}
or \cite{OSB99}. This allows us to define the set of all convolution ambiguities precisely in terms of their
factorization ambiguities in the $z-$domain, see \figref{fig:zerorepconv}, where we denoted by $\C_L[z]$ polynomials of
degree $< L$. Note, the convolution ambiguities are described  in the root-domain therefor by a partitioning map $\Pi$ of the
roots (zeros). This brings us to the following definition.
%
\if0 % alte version If the convolution is defined over the algebraic closed field $\C$ we have a unique factorization of
every polynomial $\uX(z)\in\C[z]$ of degree $D$ in $D$ primes, given as the irreducible polynomials $z^{-1}-\zeta^{-1}$
of degree one, up to a unit $u\in\C$. If we identify $\vx\in\C^L$ with the \emph{one-sided} or \emph{unilateral
$z-$transform} or \emph{transfer function} of $\vx\in\C^N$, given by
%
\begin{align}
  \uX(z)=(\Zform \vx)(z):=\sum_{k=0}^{L-1} x_k z^{-k} = \sum_{k=F}^D x_k z^{-k},
\end{align}
%
where $D$ denotes the largest and $F$ the smallest non-zero coefficient index of $\vx$. Then, we have the unique factorization 
%
\begin{align}
  \uX(z) = x_{F} \Pro_{k=1}^{D} (z^{-1}-\zeta^{-1}_k),
\end{align}
%
where $\zeta_k$ are the roots or zeros of $\uX$. Note, this holds since the ring $\C$ is a unique and algebraically
closed factorization domain, i.e., the irreducible polynomials have degree one and are given up to a unit by
$(z^{-1}-\zeta_k^{-1})$. The convolution $\vy=\vxone*\vxtwo$ of $\vxone\in\C^{\Lone}$ and $\vxtwo\in\C^{\Ltwo}$ is given
as the inverse $z-$transform $\Zform^{-1}$  of the product $\uX_1\cdot\uX_2$ of their $z-$transforms.  Therefore we can
define the set of all convolution ambiguities precisely in terms of their factorization ambiguities in the $z-$domain,
see \figref{fig:zerorepconv}.  \fi
%

%
%%%%%%%%%%%%%%%%%%%%%%%%%%%%%%%%%%%%%%%%%%%%%%%%%%%%%%%%%%%%%%%%%%%%%%%%%%%%%%%%%%%%%%%%%%%%%%%%%%%%%%%%%%%%%%%%%%%%%%%%
\begin{definition}[Convolution Ambiguities]
  Let $L_1,L_2$ be positive integers. Then the linear convolution $*\colon\C^{\Lone}\times \C^{\Ltwo}\to \C^{\Lone+\Ltwo-1}$
  defines on the domain $\C^{\Lone}\times \C^{\Ltwo}$ a equivalence relation $\sim_*$ given by
  \begin{align}
    (\vxone, \vxtwo) \sim_* (\vtxone,\vtxtwo) :\LRA \vtxone*\vtxtwo=\vxone*\vxtwo.
  \end{align}
  For each $(\vxone,\vxtwo)$ we denote by $\uX_1(z)$ and $\uX_2(z)$ its $z-$transforms of degree $D_1$ respectively
  $D_2$.  Moreover we denote by $x_{F_1}$ respectively $x_{F_2}$ the first non-zero coefficients of $\vxone$
  respectively $\vxtwo$ and by $\{\zeta_k\}_{k=1}^{D_1+D_2}\subset\C\cup\{\infty\}$ the zeros of the product $\uX_1
  \uX_2$. Then the pair
  \begin{align}
    (\vtxone,\vtxtwo):=(\Zform^{-1}(\tuX_1), \Zform^{-1}(\tuX_2)), 
  \end{align}
  with
  \begin{align*}
    \tuX_1=x_{F_1}x_{F_2}\Pro_{k\in P} (z^{-1}-\zeta_k^{-1}) 
    \quad\text{and}\quad \tuX_2=\Pro_{k\in [D]\setminus P}  (z^{-1}-\zeta_k^{-1}),
  \end{align*}
  where $P$ is some subset of $[D]$ such that  $D-L_2+1\leq |P| \leq L_1-1$, is called a \emph{left-scaled non-trivial
  convolution ambiguity} of $(\vxone,\vxtwo)$.  The set of all \emph{convolution ambiguities} of $(\vxone,\vxtwo)$ is
  then the equivalence class defined by the finite union of the scaling equivalence classes of all left-scaled
  non-trivial convolution ambiguities given by
  \begin{align}
    [(\vxone,\vxtwo)]_*:= \bigcup_{n} [(\vtxone^{(n)},\vtxtwo^{(n)})]_m.
  \end{align}
  We will call $(\vtxone,\vtxtwo)\in [(\vxone,\vxtwo)]_*$ a \emph{scaling convolution ambiguity} or \emph{trivial
  convolution ambiguity} of $(\vxone,\vxtwo)$ if $(\vtxone,\vtxtwo)\in[(\vxone,\vxtwo)]_m$ and in all other cases a \emph{non-trivial
  convolution ambiguity} of $(\vxone,\vxtwo)$.
\end{definition}
%%%%%%%%%%%%%%%%%%%%%%%%%%%%%%%%%%%%%%%%%%%%%%%%%%%%%%%%%%%%%%%%%%%%%%%%%%%%%%%%%%%%%%%%%%%%%%%%%%%%%%%%%%%%%%%%%%%%%%%%
%
\begin{remark}
  The naming \emph{trivial} and \emph{non-trivial} is borrowed from the polynomial language, where a trivial polynomial
  is a polynomial of degree zero, represented by a scalar (unit),  and a non-trivial polynomial is given by a polynomial
  of degree greater than zero.  Hence, the factorization ambiguity of a trivial polynomial corresponds to the scaling or
  trivial convolution ambiguity and the factorization ambiguity of a non-trivial polynomial corresponds to the
  non-trivial convolution ambiguity.  We want to emphasize at this point, that the $z-$domain (polynomial) picture is
  known and used for almost a century in the engineering, control and signal processing community. Hence this
  factorization of convolutions is certainly not surprising, but by the best knowledge of the authors, not rigorous
  defined in the literature.  For a factorization of the auto-correlation in the $z-$domain see for example
  \cite{HLO80}, \cite[Sec.3.]{BS79}, and the summarizing text book about phase retrieval  \cite{Hur89}. A complete
  one-dimensional ambiguity analysis for the auto-correlation problem was recently obtained by \cite{BP15}.  A very
  similar, but not a full characterization of the ambiguities of the phase retrieval problem was given by \cite{JOH13}.
  Both works extend the results in \cite{LV11}.  Let us mentioned at last, that the non-trivial ambiguities for
  multi-dimensional convolutions are almost not existence, by the observation that multivariate polynomials, chosen
  randomly, are irreducible with probability one, i.e.,  a factorization ambiguity is then not possible, see for example
  \cite{Lan02}.  This is in contrast to a random chosen univariate polynomial, which has full degree and no
  multiplicities of the zeros (factors), and obtains therefore the maximal amount of non-trivial ambiguities, see upper
  bound in \eqref{eq:boundambi}.
\end{remark}

\paragraph{On the Combinatorics of Ambiguities}

The determination of the amount $M$ of left-scaled non-trivial convolution ambiguities of some
$(\vxone,\vxtwo)\in\C^{L_1\times L_2}$ is a
hard combinatorial problem. The reason are the multiplicity of the zeros of $\uX_1\uX_2$. If a zero $\zeta_k$ has multiplicity
$m_k\geq1$, then we have $m_k+1$ possible assignments of the $m_k$ equal $\zeta_k$ to $\uX_1$, i.e., we can choose one
of the factors $\{1,1-\zeta_k, (1-\zeta_k)^2, \dots, (1-\zeta_k)^{m_k}\}$ as long as $m_k\leq \Lone-1$. Hence, if all
zeros are equal, we only have $\min\{D+1,\Lone,\Ltwo\}$ different choices to assign zeros for $\uX_1$. Contrary, if all
zeros are distinct, then we end up with $2^D$ different zero assignments for $\uX_1$, which yields to 
\begin{align}
  \min\{D+1,\Lone,\Ltwo\}\leq M&   \leq 2^D.\label{eq:boundambi}
\end{align}
In \figref{fig:x1x2ambi} we plotted for arbitrary polynomials $\uX_1$ and $\uX_2$ their zeros in the $z-$domain, where
we assumed one common zero. Since the polynomials have finite degree, the only pole is located at the origin. Every
permutation of the zeros yields then to an ambiguity. 
%
\if0
\begin{figure}[t]
  \begin{subfigure}[h]{0.48\textwidth}
    \centering
    \def\svgwidth{0.96\columnwidth}
    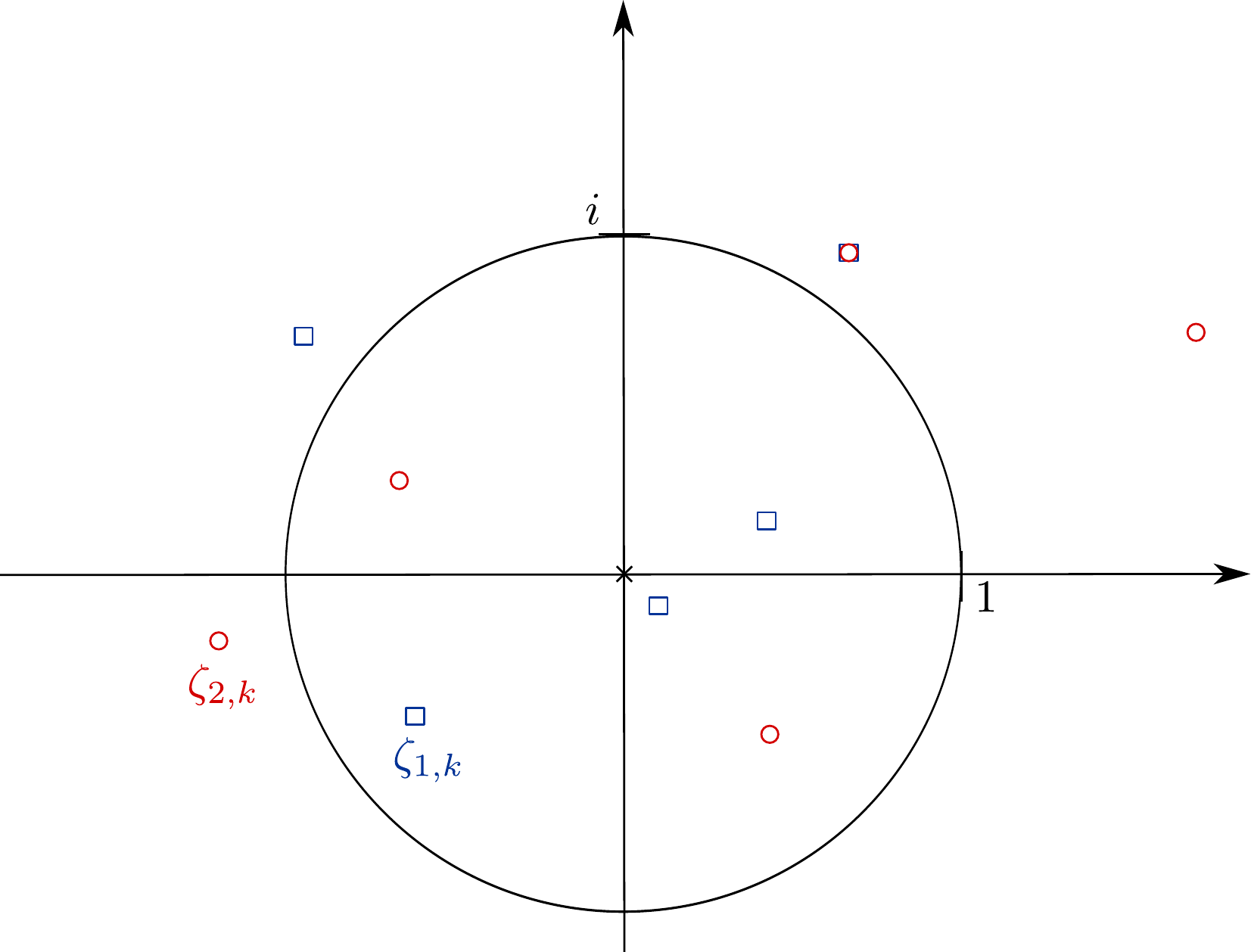
    \caption{Ambiguities for a convolution product with one common zero}\label{fig:x1x2ambi}
  \end{subfigure}
  \hspace{0.5cm}
  \begin{subfigure}[h]{0.48\textwidth}
    \def\svgwidth{0.96\columnwidth}
    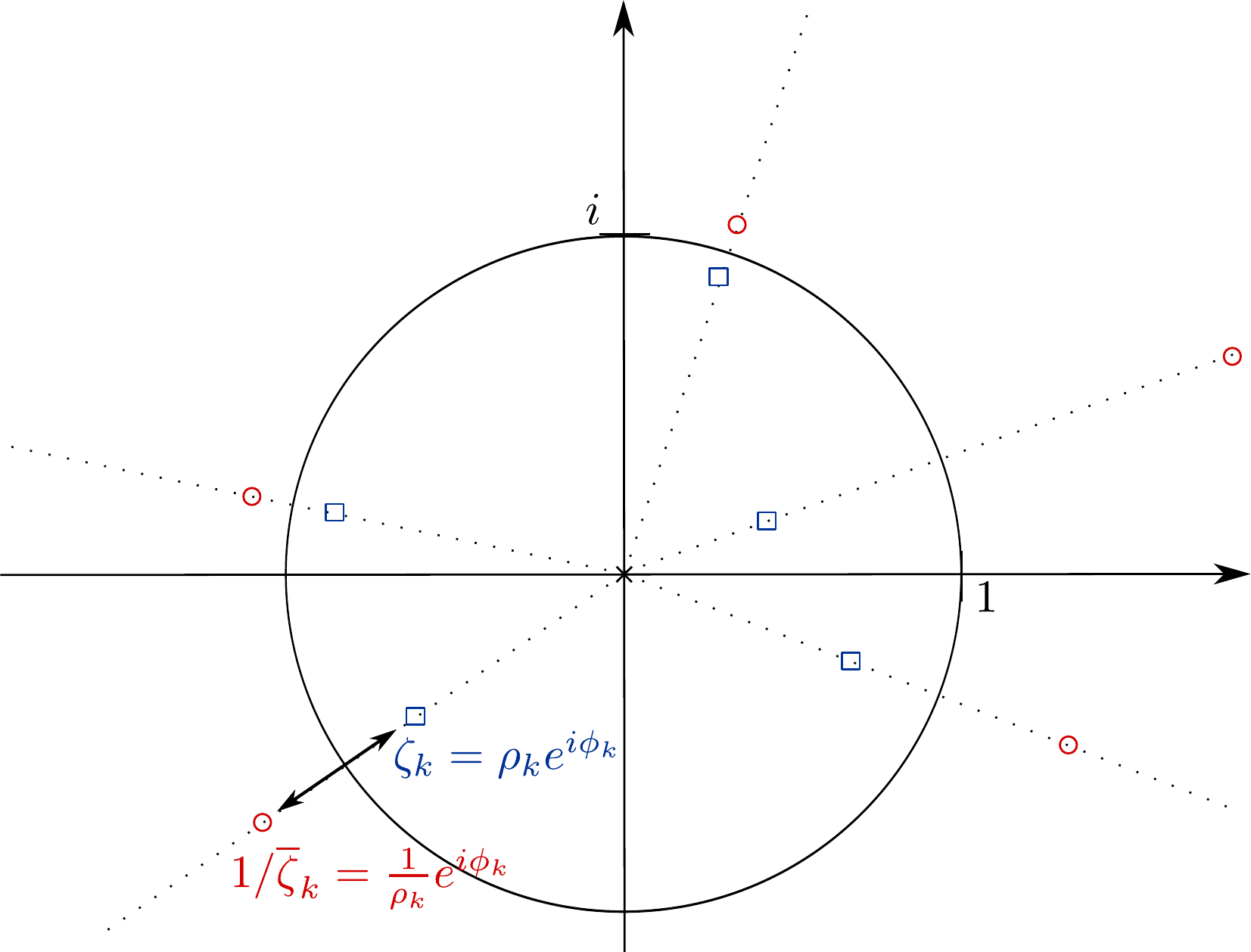
    \caption{Ambiguity for autocorrelation, by swapping the zeros of their conjugated-pairs.}\label{fig:ambiguityautcorN}
  \end{subfigure}
  \caption{Zero/pole plots for convolution products.}\label{fig:zeropoleplots}
\end{figure}
\fi

\begin{figure}[t]
  \begin{minipage}{0.48\textwidth}
    \centering
    \def\svgwidth{0.96\columnwidth}
    \input{ZeroXoneXtwo_ieee.pdf_tex}
    \caption{Ambiguities for a convolution product with one common zero}\label{fig:x1x2ambi}
  \end{minipage}
  \hspace{0.5cm}
  \begin{minipage}{0.48\textwidth}
    \def\svgwidth{0.96\columnwidth}
    \input{ZeroMin_ieee2.pdf_tex}
    \caption{Ambiguity for autocorrelation, by swapping the zeros of their conjugated-pairs.}\label{fig:ambiguityautcorN}
  \end{minipage}
%  \caption{Zero/pole plots for convolution products.}\label{fig:zeropoleplots}
\end{figure}
%%%%%%%%%%%%%%%%%%%%%%%%%%%%%%%%%%%%%%%%%%%%%%%%%%%%%%%%%%%%%%%%%%%%%%%%%%%%%%%%%%%%%%%%%%%%%%%%%%%%%%%%%%%%%%%%%%%%%%%%
\paragraph{Ambiguities of Autocorrelations} A very well investigated special case of  blind deconvolution is the
reconstruction of the signal of $\vx\in\C^N$ by its linear or aperiodic autocorrelation, see for example \cite{BP15},
given as the convolution of $\vx$ with its \emph{conjugate-time-reversal} $\cc{\Rvx}$ defined component-wise by
$(\cc{\Rvx})_k=\cc{x_{N-1-k}}$ for $k\in [N]$.  To transfer this in the $z-$domain we need to define on the polynomial
ring $\C[z]$ an \emph{involution} $(\cdot)^*$, given for any polynomial $\uX\in\C[z]$ of degree $N-1$ by
\begin{align}
  \uX^*(z):=z^{1-N} \cc{\uX(1/\cc{z})}. 
\end{align}
Then, the autocorrelation $\va=\vx*\cc{\Rvx}$ in the time-domain transfers to  $\uA=\uX\uX^*$ in the $z-$domain. 
All non-trivial correlation ambiguities are then given by assigning for the conjugated-zero-pairs $(\zeta_k,\zeta_k^*)$
of $\uA$ one zero to $\tuX$. Since we do not have more than $N-1$ different zeros for $\tuX$, we have not more than
$2^{N-1}$ different factorization ambiguities, see \figref{fig:ambiguityautcorN}. The scaling ambiguities reduce by
$\lam\cc{\lam}=1$ to a \emph{global phase scaling} $e^{i\phi}$ for any $\phi\in\R$.

\paragraph{Well-posed Blind Deconvolution Problems} To guarantee a unique solution of a deconvolution problem up to a
global scalar \cite[Def.1]{CM15a}, we have to resolve all non-trivial convolution ambiguities, which demands therefor a
unique disjoint structure on the zeros of $\uX_1$ and $\uX_2$.  The most prominent structure for a unique factorization
is given by the spectral factorization (phase retrieval) for minimum phase signals, i.e., for signals $\uX$ having all
its zeros  inside the unit circle (a zero on the unit circle has  even multiplicity and a swapping of its conjugated
pair has therefor no effect). Another structure for a blind deconvolution would be to demand that $\uX_1$ has all its
zeros inside the unit circle and $\uX_2$ has all its zeros strictly outside the unit circle. In fact, every separation
would be valid, as long as it is practical realizable for an application setup. In this spirit, the condition that
$\uX_1$ and $\uX_2$ do no have a common zero is equivalent with the statement that a unique separation is possible. This
is the weakest and hence a necessary structure  we have to demand on the zeros, which  was already exploited in
\cite{XLTK95}. However, the challenge is still to find an efficient and stable reconstruction algorithm, which have to come
with a price of further structure and constrains. But, instead of designing further constraints on the zeros, one can
also demand further measurements of $\vxone$ and $\vxtwo$. In the next section we will introduce an efficient recovery
algorithm given by a convex program with the knowledge of additional autocorrelation measurements.  

%%%%%%%%%%%%%%%%%%%%%%%%%%%%%%%%%%%%%%%%%%%%%%%%%%%%%%%%%%%%%%%%%%%%%%%%%%%%%%%%%%%%%%%%%%%%%%%%%%%%%%%%%%%%%%%%%%%%%%%%
\section{Blind Deconvolution  with Additional Autocorrelations via SDP}
%%%%%%%%%%%%%%%%%%%%%%%%%%%%%%%%%%%%%%%%%%%%%%%%%%%%%%%%%%%%%%%%%%%%%%%%%%%%%%%%%%%%%%%%%%%%%%%%%%%%%%%%%%%%%%%%%%%%%%%%

Since the \emph{autocorrelation} of a signal $\vx\in\C^N$ does not contain enough information to obtain a unique
recovery, as shown in the previous section, the idea is to use cross-correlation informations of the signal by
partitioning $\vx$ in two disjoint signals $\vx_1\in\C^{L_1}$ and $\vx_2\in\C^{L_2}$, which yield $\vx$ if stacked
together. This approach was first investigated in \cite{RDN13} and called \emph{vectorial phase retrieval}. 
The same approach was obtained independently by one of the authors in  \cite[Thm.  III.1]{JH16}, which steamed from a
generalization of a phase retrieval design in \cite[Thm.4.1.4.]{Jag16}, from three masked Fourier magnitude-measurements
in $N$ dimension, to a purely correlation measurement design between arbitrary vectors $\vxone\in\C^{\Lone}$ and
$\vxtwo\in\C^{\Ltwo}$.  To solve the phase retrieval problem via a semi-definite program (SDP), the autocorrelation or
equivalent the Fourier magnitude-measurements has to be represented as linear mappings on positive-semidefinite rank$-1$
matrices.  This is know as the \emph{lifting approach} or in mathematical terms as the \emph{tensor calculus}.  The
above partitioning of $\vx$ yields to a block structure of the positive-semidefinite matrix  
\begin{align}
  \vx\vx^*= \begin{pmatrix}\vx_1\\ \vx_2\end{pmatrix}\begin{pmatrix}\vx_1^* & \vx_2^*\end{pmatrix}=\begin{pmatrix}
    \vx_1\vx_1^* & \vx_1\vx_2^*\\
    \vx_2\vx_1^* & \vx_2\vx_2^*\end{pmatrix}.\label{eq:xx4block}
\end{align}
The linear measurement $\Alin$ are then given component-wise by the inner products with the sensing matrices
$\vA_{i,j,k}$, defined below, which correspond to the $k$th correlation components of $\vx_i *\cc{\vx^-_j}$ for
$i,j\in\{1,2\}$.  Hence, the autocorrelations and cross-correlations can be obtain from the same object $\vx\vx^*$. 
Let us define the $N\times N$ down-shift and $N\times L$ embedding matrix as 
\begin{align}
  \vT_N=
    \begin{pmatrix} 0 & \dots &  0 & 0 \\ 1 & \dots & 0 & 0\\
      \vdots & \ddots &  \vdots & \vdots\\
      0 & \dots & 1 & 0\\
    \end{pmatrix}
    \quad\text{and}\quad
  \ProjNL& =\begin{pmatrix}
    \id_{L, L}\\
    \zero_{N-L, L}\end{pmatrix},
    \label{eq:timeshiftmatrix}
\end{align}
where $\id_{L,L}$ denotes the $L\times L$ identity matrix and $\zero_{N-L,L}$ the $(N-L)\times L$ zero matrix.  Then, the $L_i\times
L_j$ rectangular shift matrices%
\footnote{Note, it holds not $\vT_{L_j,L_i}^{(k)}=\vT_{L_i,L_j}^{(k)}$ unless $L_i=L_j$, cause the involution in the
vector-time domain is $\cc{(\cdot)^-}$.} 
are defined as 
\begin{align}
  (\vT^{(k)}_{\!L_j,L_i})^T:=\Proj^T_{N,L_i}\vT^{k-L_j+1}_{N}\ProjNLj,
       \label{eq:LiLj}
\end{align}
for $k\in\{0,\dots,L_i+L_j-2\}=:[L_i+L_j-1]$, where we set $\vT^l_N:=(\vT^{-l}_N)^T$ if $l<0$. Then, the
\emph{correlation} between $\vx_i\in\C^{L_i}$ and $\vx_j\in\C^{L_j}$  is given component-wise\footnote{We use here the
vector definition and hence the time-reversal $\vx^-$ of the signal $\vx$ is a flipping of the vector coefficient indices in
$[L_i]$  and not a flipping at the origin $0$ as defined for sequences. The scalar product is given as
$\skprod{\va}{\vb}:=\sum_{k}a_k\cc{b_k}$.}   as 
\begin{align*}
  (\!\va_{i,j}\!)_k:=(\vx_i*\cc{\Rvxj})_{k}&= \skprod{\vx_i}{(\vT^{(k)}_{\!L_j\!,L_i}\!)^T\vx_j}
  =\cc{\skprod{\vx_j}{\vT^{(k)}_{\!L_j\!,L_i}\vx_i}}= \vx_j^* \vT^{(k)}_{\!L_j\!,L_i}\vx_i = \tr(\vT^{(k)}_{\!L_j\!,L_i}\vx_i\vx_j^*).
\end{align*}
Hence, this defines the linear maps $\Alin_{i,j,k}(\vX):=\tr(\vA_{i,j,k}\vX)$ for $k\in[L_i+L_j-1]$ with sensing matrices  
\begin{flalign}
   &&\begin{split}
      \vA_{1,1,k} &= \begin{pmatrix} \vT^{(k)}_{L_1,L_1} &\zero_{L_1,L_2}\\ \zero_{L_2,L_1}
        &\zero_{L_2,L_2}\end{pmatrix}  \!\!\quad,\quad k\in[2L_1-1]\end{split} && \label{eq:a11}\\
   && \begin{split}
       \vA_{2,2,k} &= \begin{pmatrix} \zero_{L_1,L_1} & \zero_{L_1,L_2}\\ \zero_{L_2,L_1} & \vT^{(k)}_{L_2,L_2} \end{pmatrix} 
    \! \!\quad,\quad k\in [2L_2-1] \end{split} &&\\ 
   && \begin{split} 
     \vA_{1,2,k} & = 
     \begin{pmatrix} \zero_{L_1,L_1} &  \zero_{L_2,L_1}\\ \vT^{(k)}_{L_2,L_1}& \zero_{L_2,L_2}\end{pmatrix} 
   \!\!\quad,\quad k\in [L_1+L_2-1] \end{split}&&\\
  && \begin{split}  
    \vA_{2,1,k}& = \begin{pmatrix} \zero_{L_1,L_1} & \vT^{(k)}_{L_1,L_2}\\ \zero_{L_2,L_1} & \zero_{L_2,L_2}\end{pmatrix} 
  \!\!\quad,\quad k\in [L_1+L_2-1].\end{split}&& 
   \label{eq:a21}
\end{flalign}
Stacking all the $\Alin_{i,j}$ together gives the measurement map $\Alin$. Hence, the $4N-4$ linear measurements are 
\begin{align}
   \vb:= \Alin(\vx\vx^*)
      =\begin{pmatrix} 
        \Alin_{1,1}(\vx\vx^*)\\ 
        \Alin_{2,2}(\vx\vx^*)\\
        \Alin_{1,2}(\vx\vx^*)\\
        \Alin_{2,1}(\vx\vx^*)
      \end{pmatrix}
    =
    \begin{pmatrix} 
      \vxone*\cc{\Rvxone} \\ 
      \vxtwo*\cc{\Rvxtwo}\\
      \vxone*\cc{\Rvxtwo} \\
      \vxtwo*\cc{\Rvxone}
    \end{pmatrix}=:\begin{pmatrix}\va_{1,1}\\\va_{2,2}\\\va_{1,2}\\\va_{2,1}\end{pmatrix}.\label{eq:3Nb}
\end{align}
Note, since the cross-correlation $\va_{1,2}$ is the conjugate-time-reversal of $\va_{2,1}$, i.e.,
$\va_{1,2}=\cc{\va_{2,1}^{-}}$, we only need $3N-3$ correlation measurements to determine $\vb$.

%%%%%%%%%%%%%%%%%%%%%%%%%%%%%%%%%%%%%%%%%%%%%%%%%%%%%%%%%%%%%%%%%%%%%%%%%%%%%%%%%%%%%%%%%%%%%%%%%%%%%%%%%%%%%%%%%%%%%%%%
\subsection{Unique Factorization of Self-Reciprocal Polynomials}
%%%%%%%%%%%%%%%%%%%%%%%%%%%%%%%%%%%%%%%%%%%%%%%%%%%%%%%%%%%%%%%%%%%%%%%%%%%%%%%%%%%%%%%%%%%%%%%%%%%%%%%%%%%%%%%%%%%%%%%%

To prove our main result in \thmref{thm:4correlation} we need a unique factorization of self-reciprocal polynomials in
irreducible self-reciprocal polynomials, where we call a polynomial $\uX$ \emph{self-inversive} if $\uX^*=e^{i\alp}\uX$
for some $\alp\in[0,2\pi)$ and \emph{self-reciprocal}
\footnote{In the literature there also called conjugate-self-reciprocal to distinguish them from the real case $\K=\R$.
  For $\K=\R$ or $\K=\Z$ they are also called \emph{palindromic polynomials} or simply \emph{palindromes} (Coding
Theory).}
if $\alp=0$, see for example \cite{Vie15} and reference therein. The term self-reciprocal refers to the
\emph{conjugate-symmetry} of the coefficients, given by
\begin{align}
  \vx=\cc{\Rvx}\in\C^N,
\end{align}
which can be used as the definition  of a self-reciprocal polynomial by its coefficients.
In fact, it was shown by some of the authors in \cite{WJ14b} and \cite{WJP15}, that the autocorrelation of conjugate-symmetric
vectors is stable up to a global sign.
As for the unique factorization \eqref{eq:uniquefacX} of any polynomial $\uX\in\C[z]$ of degree $D\geq1$ in $D$
irreducible polynomials (primes) $\uP_k(z)=1-\zeta_kz^{-1}$,  up to a unit $u\in\C\setminus\{0\}$, we can ask for a unique
factorization of any self-reciprocal polynomial $\uS\in\C[z]$ in irreducible self-reciprocal polynomials $\uS_k$, i.e.,
$\uS_k$ can not be further factored in self-reciprocal polynomials of smaller degree. 
To see this, we first use the definition of a self-reciprocal factor $\uS$ of degree $D$, which demands that each zero
$\zeta$ comes with its conjugate-inverse pair $1/\cc{\zeta}=:\zeta^*$. If $\zeta$ lies on the unit circle, then we have $\zeta=\zeta^*$
and the multiplicities of these zeros can be even or odd. Let us assume we have $T$ zeros on the unit circle, then we
get the factorization
\begin{align*}
  \uS(z)= u\Pro_{k=1}^{\frac{D-T}{2}}(1-\zeta_k z^{-1}) (1-\cc{\zeta_k^{-1}} z^{-1}) \Pro_{k=D-T+1}^D (1-\zeta_k
  z^{-1}), 
\end{align*}
where the phase $\alp$ of the unit $u\in\C$ is determined by the phases $\alp_k$ of the conjugate-inverse zeros. To see this
we derive
\begin{align}
  \uS^*\!(z) &= z^{-D} \cc{u} \Pro_{k=1}^{\frac{D-T}{2}} (1-\cc{\zeta_k} z) (1-\zeta_k^{-1} z) \Pro_{k=D-T+1}^D
  (1-\cc{\zeta_k} z)\notag\\
  &=  \cc{u} \Pro_{k} \frac{\cc{\zeta_k}}{\zeta_k}(1-\frac{1}{\cc{\zeta_k}}z^{-1})(1-\zeta_k z^{-1})
  \Pro_k (-\cc{\zeta_k})(1-\frac{1}{\cc{\zeta_k}} z^{-1}).\notag
\intertext{If we set for the zeros $\zeta_k=\rho_k e^{i\alp_k}$ and unit $u=\rho e^{i\alp}$ we get}
   &=  e^{-i(2\alp + \sum_{k=1}^{D} \alp_k -T\pi)} \uS(z)\overset{!}{=} \uS(z).\notag
\end{align}
Hence, it must hold for the phase $\alp=(T\pi-\sum_{k=1}^{D}\alp_k)/2$.   Moreover, for every prime $\uP_k$ of $\uS$
also $\uP_k^*$ is a prime of $\uS$. Hence, if $\uP_k\not=\uP_k^*$ then $\uS_k:=\uP_k\uP_k^*$ is a self-reciprocal factor
of $\uS$ of degree two. If $\uP_k=\uP_k^*$, then $\uS_k:=\uP_k$ is already a self-reciprocal factor of $\uS$ of degree
one. However, the conjugate-inverse factor pairs $(1-\zeta_kz^{-1})(1-\cc{\zeta_k}^{-1}z^{-1})$ are not self-reciprocal,
but self-inversive. We have to scale them with $e^{-i\alp_k}$ to obtain a self-reciprocal factor $\uS_k:=\uP_k\uP_k^*$,
i.e., we have to set $\uP_k(z):=\rho_k^{-1/2}(1-\zeta_k z^{-1})$.  Similar, for the primes on the unit circle, we set
$\uS_k(z):=e^{-i(\pi+\alp_k)/2}(1-e^{i\alp_k} z^{-1})$.  Hence, we can  write $\uS$ as a factorization of
\emph{irreducible self-reciprocal polynomials} $\uS_k$, i.e., self-reciprocal polynomials which are not further
factored in self-reciprocal polynomials of smaller degree,
\begin{align}
  \uS=\Pro_{k=1}^{\frac{D-T}{2}} \underbrace{\uP_k\uP_k^*}_{=\uS_k} \Pro_{k=D-T+1}^D \uS_k.
\end{align}
Let us define the greatest self-reciprocal factor/divisor (GSD).
%
%%%%%%%%%%%%% DEFINITION %%%%%%%%%%%%%%%%%%%%%%%%%%%%%%%%%%%%%%%%%%%%%%%%%%%%%%%%%%%%%%%%%%%%%%%%%%%%%%%%%%%%%%%%%%%%%%%%%%%
\begin{definition}[Greatest Self-Reciprocal Divisor]
  Let $\uX\in\C[z]$ be a non-zero polynomial. Then the \emph{greatest self-reciprocal divisor} 
  $\uS$ of $\uX$ is the self-reciprocal factor with largest degree. It is unique up to a real-valued trivial factor 
  $c\in\R$. 
\end{definition}
%%%%%%%%%%%%%%%%%%%%%%%%%%%%%%%%%%%%%%%%%%%%%%%%%%%%%%%%%%%%%%%%%%%%%%%%%%%%%%%%%%%%%%%%%%%%%%%%%%%%%%%%%%%%%%%%%%%%%%%%%%
%
Let us denote by $\uC\mid\uX$ that $\uC$ is a factor/divisor of the polynomial $\uX$ and by $\uC\nmid\uX$ that $\uC$ is
not.  Then $\uC\mid\uX$ and $\uC\mid\uY$ is equivalent to the assertion that $\uC$ is a common factor of $\uX$ and
$\uY$.  For any polynomial $\uX\in\C[z]$,  which factors in $\uX=\uS\uR$, it holds
\begin{align}
  \uS \text{ self-reciprocal }\quad \RA \quad \uS\mid\uX \text{ and }\uS\mid\uX^*,\label{eq:sfiscf}
\end{align}
since it holds by the self-reciprocal property of $\uS$  
\begin{align}
    \uX^*=\uS^*\uR^*=\uS\uR^*,
\end{align}
which proofs that $\uS\mid\uX$ and $\uS\mid\uX^*$.  For the reverse we can only show this for the \emph{greatest common
divisor} (GCD).
%
%%%%%%%%%%%%%% LEMMA %%%%%%%%%%%%%%%%%%%%%%%%%%%%%%%%%%%%%%%%%%%%%%%%%%%%%%%%%%%%%%%%%%%%%
\begin{lemma}\label{lem:gsfgcf}
  For $\uX\in\C[z]$ it holds
  \begin{align}
   \uG \text{ is GSD of }\uX \quad \LRA \quad \uG \text{ is GCD of  }\uX \text{ and }\uX^*.
  \end{align}
\end{lemma}
%%%%%%%%%%%%%% LEMMA: ENDE %%%%%%%%%%%%%%%%%%%%%%%%%%%%%%%%%%%%%%%%%%%%%%%%%%%%%%%%%%%%%%%%
%
\begin{proof}
  The ``$\RA$'' follows from \eqref{eq:sfiscf} since a GSD is trivially also a self-reciprocal factor of $\uX$ and
  therfor a factor of $\uX^*$. To see
  the other direction, we denote by $\uG$ the GCD of $\uX$ and $\uX^*$, which factorize as 
  \begin{align}
    \uX=\uG\uR \quad\text{and}\quad \uX^*=\uG\uQ\label{eq:xxstar},
  \end{align}
  where $\uR$ and $\uQ$ are the co-factors of $\uX$ respectively $\uX^*$. Then we get
  \begin{align}
    \uX^* =\uG^*\uR^* = \uG\uQ\label{eq:uxgrgq}.
  \end{align}
  Let us assume $\uG$ is not self-reciprocal, i.e., $\uG\not=\uG^*$, then we can still factorize $\uG$, as any
  polynomial,  in the \emph{greatest self-reciprocal factor} $\uS$ and a \emph{non-self-reciprocal factor} $\uN$. Note,
  it might also hold the trivial case $0\not=\uS=c\in\R$. Moreover, if the multiplicity of at least one zero in $\uS$,
  not lying on the unit circle, is larger than one, then $\uN$ might contain this zero (if the corresponding
  conjugate-inverse zero is missing in $\uG$). It is clear, that $\uN$ can not contain more than $(D-T)/2$ such isolated
  factors, lets call the product of all them $\uI_1$ and $\uI_2$ resp. $\uN_2$ the co-factors, i.e., $\uS=\uI_1\uI_2$
  and $\uN=\uI_1\uN_2$. Hence,
  $\uI_1$ is the GCD of $\uS$ and $\uN$. Then \eqref{eq:xxstar} becomes 
  \begin{align}
    \uX=\uS\uN\uR \quad \text{and}\quad  \uX^*=\uS\uN^*\uR^*\label{eq:snq},
  \end{align}
  which yields to
  \begin{align}
    \uG\mid\uX^* \LRA \uS\uN \mid \uS\uN^*\uR^* \LRA \uN\mid\uN^*\uR^* 
    \RA \uN \mid \uI_1^* \uN_2^*\uR^*.
  \end{align}
  Then $\uI_1^*\nmid\uN$,  since, if any factor $\tuI_1^*\subset\uI_1^*$ would be a factor of $\uN$, then
  also $\tuI_1^*\mid \uN^*$ and hence $\tuI_1\mid \uN$ and therefore $\tuI_1\tuI_1^*\mid \uN$,
  which would be a  non-trivial self-reciprocal factor and contradicts the definition of $\uN$. By the same reason
  $\uN_2^*\nmid\uN$ since any non-trivial factor of $\uN_2^*$ would result in a non-trivial self-reciprocal factor of
  $\uN$ which is again a contradiction. Hence $\uN\mid\uR^*$, i.e., we have $\uR^*=\uN\uT$ which yields to 
  \begin{align}
    \uR=(\uR^*)^*=(\uN\uT)^* = \uN^* \uT^*.
  \end{align}
  On the other hand it holds also 
  \begin{align*}
    \uG\uQ \overset{\eqref{eq:uxgrgq}}{=}\uX^*\overset{\eqref{eq:snq}}{=}\uS\uN^*\uR^*
       =  \uS\uN^*\uN\uT =\uG\uN^*\uT  \quad \RA\quad \uQ=\uN^*\uT.
  \end{align*}
  Hence $\uN^*\mid\uR$ and $\uN^*\mid\uQ$ and by \eqref{eq:xxstar} also $\uN^*\mid\uX$ and $\uN^*\mid\uX^*$, which is a
  contradiction, since $\uG$ is the GCD of $\uX$ and $\uX^*$. Hence the assumption is wrong and it must hold
  $\uG=\uG^*$. To see that $\uG$ is also the GSD, assume $\tuG$ would be self-reciprocal and contain $\uG$ as factor,
  then $\tuG$ would be by \eqref{eq:sfiscf} a common factor which is greater then $\uG$ and hence contradicts again with
  $\uG$ to be the GCD. 
\end{proof}
%
\if0
\begin{remark}
  Note, the other direction in \eqref{eq:sfiscf} is not always true. Take for example the polynomial $\uX=\uS\uS^*\uR$
  with $\uS=1-\zeta z^{-1}$ with $|\zeta|\not=1$. Then $\uX^*=\uS^*\uS\uR^*$ and clearly $\uS$ is a common factor but it
  is not even self-inversive since $\uS^*=z^{-1} (1-\cc{\zeta_k} z)= -\cc{\zeta}
  (1-\cc{\zeta^{-1}}z^{-1})\not=e^{i\alp}(1-\zeta z^{-1})$ for any $\alp\in\R$. 
\end{remark}
\fi % kein remark in conference

Let us define the \emph{anti-self-reciprocal polynomial} $\uA$ by the property $\uA=-\uA^*$, where $i$ is the
trivial\footnote{Actually, also $ic$ for any $c\in\R$ would be a trivial anti-self-reciprocal factor. But since we are
interested in the factorization of a anti-self-reciprocal polynomial in a self-reciprocal $\uS$ and a trivial
anti-self-reciprocal $i$, we can assign the $c$ either to $i$ or to $\uS$.} anti-self-reciprocal factor. Hence, for any
self-reciprocal factor $\uS$ we get by $\uA=i\uS$ an anti-self-reciprocal factor. Hence, if we factorize $\uX$ in the
GSD $\uG$ and the co-factor $\uR$, we obtain with the identity $-i\cdot i=1$ the factorization
\begin{align}
  \uX=i\uG\tuR,
\end{align}
where $\tuR=-i\uR$ does not contain non-trivial self-reciprocal or anti-self-reciprocal factors.  With this we can show the
following result.
%
%%%%%%%%% LEMMA %%%%%%%%%%%%%%%%%%%%%%%%%%%%%%%%%%%%%%%%%%%%%%%%%%%%%%%%%%%%%%%%%%%%%%%%%%%%%%%%%%%%%%%%%%%%%%%%%%%%%%%%%%
\begin{lemma}\label{lem:antiselfinversive}
  Let $\uX\in\C[z]$ be a polynomial of order $L\geq 0$ with no infinite zeros and let $\uG,\uR$ be polynomials with
  $\uX=\uG\uR$, where $\uG$ is the GSD of degree $D\leq L$ and $\uR$ its co-factor. Then the only non-zero polynomial
  $\uH$ of order $\leq L$, which yields to a anti-self-reciprocal product $\uX\uH^*$, i.e., fulfills
  \begin{align}
    \uX\uH^* + \uX^*\uH=0\label{eq:xh}
  \end{align}
  is given by $\uH=i\uR\uS$ for any self-reciprocal polynomial $\uS$ of degree $\leq D$.
\end{lemma}
%%%%%%%%%%%%%%%%%%%%%%%%%%%%%%%%%%%%%%%%%%%%%%%%%%%%%%%%%%%%%%%%%%%%%%%%%%%%%%%%%%%%%%%%%%%%%%%%%%%%%%%%%%%%%%%%%%%%%%%%%%
%
\begin{proof}
  Since $\uX$ factors in the GSD $\uG$ and the co-factor $\uR$ we have by \lemref{lem:gsfgcf} that $\uG$ is the
  GCD of $\uX$ and $\uX^*$, which gives $\uX^*=\uG\uR^*$.
  Inserting this factorization in \eqref{eq:xh} yields to 
  \begin{align}
    \uG(\uR\uH^* +\uR^*\uH)=0 \LRA \uR\uH^* + \uR^*\uH=0.\label{eq:coprimerrstar}
  \end{align}
  Since $\uR$ and $\uR^*$ do not have a common factor by definition of $\uG$, but have degree $L-D$, which is less or
  equal then $\uH$ and $\uH^*$ (degree $\leq L$), the only solution is of the form 
  \begin{align}
    \uH=i\uR \uS,\label{eq:irs}
  \end{align}
  where $\uS$ is any self-reciprocal polynomial of degree $\leq D$. 
\end{proof}
\begin{remark}
  This result \eqref{eq:irs} can be seen as a special case of the Sylvester Criterion for the  polynomials
  $\uA=\uX_1$ and $\uB=\uX_2^*$ in \eqref{eq:xh}, where $\uG$ is the GCD of $\uA$ and $\uB$. Hence $\uH$ and $\uH^*$
  have GCD $\uS$ of degree $D$, which must be by \lemref{lem:gsfgcf} the GSD with  the
  co-factors  $i\uR$ respectively $-i\uR^*$, see \appref{app:sylvester}. 
\end{remark}

\subsection{Main Result}
Let us denote by $\C^{L}_{0,0}:=\set{\vx\in\C^L}{x_0\not=0\not=x_{L-1}}$.
Then in \cite[Thm.III.1]{JH16} and extended by the author (purely deterministic) it holds the following theorem. 
%

%%%% good so far
%
\begin{theorem}\label{thm:4correlation}
  Let $L_1,L_2$ be positive integers and $\vx_1\in\C_{0,0}^{L_1}, \vx_2\in\C_{0,0}^{L_2}$ such that their $z-$transforms
  $\uX_1(z)$ and $\uX_2(z)$ do not have any common factors. Then $\vx^T=(\vx_1^T,\vx_2^T)\in\C^N$ with $N=L_1+L_2$ can
  be recovered uniquely up to global phase from the measurement $\vb\in\C^{4N-4}$ defined  in \eqref{eq:3Nb} by solving
  the feasible convex program
  \begin{align}
    \find \vX\in\C^{N\times N}\quad\text{s.t.}\quad \begin{split}\Alin(\vX)=\vb\\ \vX\mgeq 0\end{split} \label{eq:fcp3N}
  \end{align}
  which has $\vX^{\#}=\vx\vx^*$ as the unique solution. 
\end{theorem}

\begin{remark}
  The condition that the first and last coefficient does not vanish, guarantee that $\uX_i$ and $\uX_i^*$ have no  zeros
  at the origin and are both of degree $L_i$. Since the correlation is conjugate-symmetric, we only need to measure one cross
  correlation, since we have $\vx_1*\cc{\vx^-_2}=(\cc{\vx_2*\cc{\vx^-_1}})^{-}$. Hence we can omit the last $N-1$
  measurements in $\vb$ and achieve recovery from only $3N-3$ measurements. In fact, if we set $\vtx_2=\cc{\Rvxtwo}$ and
  demand $\uX_1$ and $\tuX_2^*=\uX_2$  to be co-prime, then the theorem gives recovery up to global phase of $\vx_1$ and
  $\vtx_2$ from  its convolution
  \begin{align}
    \vx_1*\vtx_2
  \end{align}
  by knowing additionally the auto-correlations $\va_{1,1}$ and $\vta_{2,2}$, since it holds by conjugate-symmetry of the
  autocorrelations, that
  $\va_{2,2}=\vx_2*\cc{\Rvxtwo}=\cc{\vtx_2^-}*\vtx_2 =\cc{(\vtx_2*\cc{\vtx_2^-})^-}=\vtx_2*\cc{\vtx_2^-}=\vta_{2,2}$. 
\end{remark}
%

\if0 %alt
\begin{theorem}\label{thm:4correlation}
  %
  Let $\vx_1\in\C_{0,0}^{L_1}$ and $\vx_2\in\C_{0,0}^{L_2}$ such that the $z-$transforms  $\uX_1(z)$ and $\uX_2(z)$ do not have any
  common factors. Then $\vx^T=(\vx_1^T,\vx_2^T)\in\C^N$ with $N=L_1+L_2$ can be recovered
  uniquely up to global phase from the measurement $\vb$ defined  in \eqref{eq:4maskmeasurements} by
  solving the feasible convex program
  %
  \begin{align}
    \find \vX\in\C^{N\times N}\quad\text{s.t.}\quad \begin{split}\Alin(\vX)=\begin{pmatrix}\vb\\
        \vR\cc{\vb_3}\end{pmatrix}\\ \vX\geq 0\end{split} \label{eq:fcp3N}
  \end{align}
  %
  which has $\vX^{\#}=\vx\vx^*$ as the unique solution. 
\end{theorem}
%
\begin{remark}
  Since $L_1,L_2$ are known, also $\vx_1$ and $\vx_2$ can be recovered up to global phase.  The condition that the first
  and last coefficient does not vanish, establish that $\uX_i$ and $\uX_i^*$ have no zeros at infinity and are both of
  degree $L_i$. Since the correlation is conjugate-symmetric, we only need to measure one cross correlation, since we
  have $\vx_1\triangle\vx_2=\vR(\cc{\vx_2\triangle\vx_1})$. Hence we can omit the last $N-1$ measurements in $\vb$ and
  achieve recovery from only $3N-3$ measurements. 
\end{remark}
%
\fi

\begin{proof}[Proof of Theorem \ref{thm:4correlation}]
In \cite{JH16} the authors could show that the feasible convex program is solvable by constructing a unique dual certificate
which implies to show the uniqueness condition only on the tangent space $T_{\vx}$ of $\vx\vx^*$, for a detailed proof
see \appref{sec:prooflemma}.
\begin{lemma}\label{lem:dualcertifrank2}
  The feasible convex problem given in \eqref{eq:fcp3N} has the unique  solution $\vX^{\#}=\vx\vx^*$ if the following
  conditions are satisfied
  \begin{enumerate}[1)]
    \item There exists a dual certificate $\vW\in \Alin^*$ such that \label{enu:dualcertif}
      \begin{enumerate}[(i)]
        \item $\vW\vx=\zero$\label{enu:zero1}
        \item $\rank(\vW) =N-1$ \label{enu:rank1}
        \item $\vW\mgeq 0$ \label{enu:positive1}
      \end{enumerate}
      \item For all $ \vH\in T_{\vx}:=\set{\vx\vh^* + \vh\vx^*}{\vh\in\C^N}$ it holds\label{enu:nullset}
        \begin{align} \Alin(\vH)=\zero \RA \vH=\zero.
        \end{align}
  \end{enumerate}
\end{lemma}
Indeed, the conditions in \enuref{enu:dualcertif} are satisfied  for the dual certificate  $\vW:=\vS^*\vS$, where $\vS$
is the $N\times N$ Sylvester matrix of the polynomials $-z^{L_1}\uX_1$ and $z^{L_2}\uX_2$ given in
\eqref{eq:sylvestermatrix}.  To see the first two conditions \enuref{enu:zero1} and \enuref{enu:rank1}, we use the
Sylvester \thmref{thm:sylvesterrank} in \appref{app:sylvester}, which states, that the only non-zero vector in the
one-dimensional nullspace of the Sylvester matrix $\vS$ is given by $\vx^T=(\vx_1^T,\vx^T_2)$  (up to scalar), i.e., by
\eqref{eq:sylvesternullspace} we have
\begin{align}
  \vS\vx= \begin{pmatrix}\vx_2*\vx_1 - \vx_1*\vx_2 \\ 0\end{pmatrix} =\zero_N\label{eq:crossconvolutionzero},
\end{align}
where the difference of the cross-convolutions vanishes due to the commutation property of the convolution. Since the
dimension of the nullspace is $1$ we have $\rank(\vS)=N-1$, which shows  \enuref{enu:rank1} with
$\rank(\vS^*\vS)=\rank(\vS)$.  The positive-semi-definitness in \enuref{enu:positive1} is given by definition of
$\vW=\vS^*\vS$, since for any matrix $\vA\in\C^{N\times N}$ it holds $\vA^*\vA\mgeq 0$. To see that $\vW=\vS^*\vS$ is in
the range of $\Alin^*$, we have to set $\vlam$ accordingly , since $\vS^*\vS$ corresponds to the correlations of
the $\vx_i$ and $\vx_j$  by \eqref{eq:sylvesterbanded} to \eqref{eq:appvlam}, it turns out that
the $\vlam$ can be decompose in terms of our four measurements, see  \appref{app:dualcerti}.

Hence, it remains to show the uniqueness condition \enuref{enu:nullset} in \lemref{lem:dualcertifrank2}, which is the
new result in this work. For that we have to show for any $\vH\in T_{\vx}$ given by $\vH=\vx\vh^*+\vh\vx^*$ for some
$\vh\in\C^N$ that it follows $\vH=\zero$ from
\begin{align}
  \Alin(\vH)=\Alin(\vx\vh^*) + \Alin(\vh\vx^*)=\zero.\label{eq:hzeroalin}
\end{align}
Here $\vH$ produce a sum of different correlations which have to vanish.  As we split $\vx$ in $\vx_1$ and $\vx_2$  we can also split
$\vh$ in $\vh_1$ and $\vh_2$.  Then we can use the block structure in $\vx\vh^*$ and $\vh\vx^*$ to split condition
\eqref{eq:hzeroalin} in
\begin{align}
  \Alin(\vH)= 
    \begin{pmatrix} 
    \vx_1 * \cc{\vh_1^-}\\
    \vx_2 *\cc{\vh_2^-}\\
    \vx_1 *\cc{\vh_2^-}\\
    \vx_2 *\cc{\vh_1^-}\\
    \end{pmatrix}
    +
    \begin{pmatrix}  
      \cc{\vx_1^-} *\vh_1\\
      \cc{\vx_2^-} *\vh_2\\ 
      \cc{\vx_2^-} *\vh_1   \\
      \cc{\vx_1^-}*\vh_2 
    \end{pmatrix}=\zero\label{eq:corrsums}.
\end{align}
Let us translate the four equations in \eqref{eq:corrsums} to the $z-$domain: \seeintern{\eqref{eq:polycorr}}
\begin{align}
  \uX_1\uH_1^* +\uX_1^*\uH_1 &=0\label{eq:x1h1}\\
  \uX_2\uH_2^* +\uX_2^*\uH_2 &=0\label{eq:x2h2}\\
  \uX_1\uH_2^* +\uX_2^*\uH_1 &=0\label{eq:x1h2},
\end{align}
where we ommited the last one, which is redundant to \eqref{eq:x1h2}.
Let us assume, $\uX_1=\uG_1\uR_1$ and $\uX_2=\uG_2 \uR_2$ where $\uG_1$ and $\uG_2$ are the GSDs of $\uX_1$ respectively
$\uX_2$ and $\uR_1$ and $\uR_2$ their co-factors, then we can find by \lemref{lem:antiselfinversive} self-reciprocal factors $\uS_1$ and $\uS_2$ such that
$\uH_1=i\uR_1\uS_1$ and $\uH_2=i\uR_2\uS_2$ which are the only solutions for \eqref{eq:x1h1} and \eqref{eq:x2h2}.
But, then it follows for the second equation \eqref{eq:x1h2}, 
\begin{align}
  0 &= \uX_1 \uH_2^* + \uX_2^*\uH_1  = -i\uG_1 \uR_1 \uR_2^* \uS_2 + i\uG_2 \uR_2^* \uR_1 \uS_1 
  \quad\LRA \quad  \uG_1 \uR_1 \uR_2^* \uS_2= \uG_2 \uR_1 \uR_2^* \uS_1\quad 
  \LRA \quad  \uG_1\uS_2 = \uG_2\uS_1\label{eq:crosszcondition}
\end{align}
%
%
\if0 % wrong!!
Here, only the self-reciprocal polynomials  $\uS_1$ and $\uS_2$ of degree $\leq D_1$ respectively $D_2$ can be chosen
freely.  Hence $\uS_1$ and $\uS_2$ can not be the inversive of $\uG_1\uR_1\uR_2$ respectively $\uG_2\uR_2^*\uR_1^*$,
since this would imply infinite degree for $\uS_1$ and $\uS_2$.  Since $\uX_1$ and $\uX_2$ do not have a common factor
also $\uG_1$ and $\uG_2$ do not have a common factor.  Hence $\uS_1$ must contain $\uG_1$ and $\uS_2$ must contain
$\uG_2$. But since $\uR_1$ and $\uR_2$ are not self-reciprocal, we can not add them to the self-reciprocal polynomials
$\uS_1$ respectively $\uS_2$ unless they are trivial self-reciprocal, i.e., $\uR_1=c_1,\uR_2=c_2$, which would give
$\uS_1=\uG_1$ and $\uS_2=\uG_2$ as only solution of \eqref{eq:crosszcondition} (only if $D_1=D_2$).  This would imply
$\uH_1=ic_1\uG_1=i\uX_1$ and $\uH_2=ic_2\uG_2=i\uX_2$ which gives $\vh=i\vx$ and result in
\fi
%
Here, only the self-reciprocal polynomials  $\uS_1$ and $\uS_2$ of degree $\leq D_1$ respectively $\leq D_2$ can be
chosen freely. Since $\uX_1$ and $\uX_2$ do not have a common factor also $\uG_1$ and $\uG_2$ do not have a common
factor and the $D\times D$ Sylvester matrix $\vS_{z^{-1}\uG_1,-z^{-1}\uG_2}$ has rank $D_1+D_2-1=D-1$ and again as in
\eqref{eq:sylvesterpoly} the only solutions for \eqref{eq:crosszcondition} are given by $\uS_1=c\uG_1$ respectively
$\uS_2=c\uG_2$ for some $c\in\R$ (note, $\uS_1$ and $\uS_2$ must be self-reciprocal, hence only real units). This
would imply $\uH_1=ic\uR_1\uG_1=ic\uX_1$ and $\uH_2=ic\uR_2\uG_2=ic\uX_2$, which gives $\vh=ic\vx$ and result in
\begin{align}
  \vH=\vx\vh^*+\vh\vx^*=-i c\vx \vx^* + ic\vx\vx^*=\zero. \label{eq:hzerovx}
\end{align}
\end{proof}

\begin{remark}
  To guarantee a unique solution of \eqref{eq:fcp3N} we need both auto-correlations, since only then we obtain both
  constraints \eqref{eq:x1h1} and \eqref{eq:x2h2}, which yielding to the constraints $\uH_1=i\uR_1^*\uS_1$ and
  $\uH_2=i\uR_2^*\uS_2$. If one of them is missing, we can construct by \eqref{eq:crosszcondition} non-zero $\vH$'s satisfying
  \eqref{eq:hzeroalin}, and hence violating the uniqueness in \enuref{enu:nullset}. 
\end{remark}

% Simulations

%%%%%%%%%%%%%%%%%%%%%%%%%%%%%%%%%%%%%%%%%%%%%%%%%%%%%%%%%%%%
\section{Simulation and Robustness} 
%%%%%%%%%%%%%%%%%%%%%%%%%%%%%%%%%%%%%%%%%%%%%%%%%%%%%%%%%%%%

If we obtain only noisy correlation measurements, i.e., disturbed by the noise vectors $\vn_{1,1},\vn_{2,2},\vn_{1,2}$ as
\begin{align}
  \vb=\vy+\vn=\begin{pmatrix}
    \va_{1,1}\\ \va_{2,2} \\ \va_{1,2}\\ \va_{2,1} \end{pmatrix}
  +\begin{pmatrix}\vn_{1,1} \\ \vn_{2,2} \\ \vn_{1,2}\\ \cc{\vn_{1,2}^-}\end{pmatrix},
\end{align}
we can search for the least-square solutions in \eqref{eq:fcp3N}, given as
\begin{align}
  \vX^{\#}:=\argmin_{\vX\mgeq 0} \Norm{\vb-\Alin(\vX)}_2^2.
\end{align}
Extracting form $\vX^{\#}$ via SVD the best rank$-1$ approximation $\vx^{\#}$ gives the normalized MSE of the reconstruction
\begin{align}
  \text{MSE}:=\min_{\phi\in[0,2\pi)} \frac{\Norm{\vx-e^{i\phi}\vx^{\#}}^2}{\Norm{\vx}^2}.
\end{align}
We plotted the normalized MSE in \figref{fig:simulation} over the received SNR (rSNR), given by
\begin{align}
  \text{rSNR}:= \frac{\Expect{\Norm{\vy}_2^2}}{\Expect{\Norm{\vn}_2^2}}.
\end{align}
Since, the noise is i.i.d. Gaussian we get for rSNR $:=\Expect{\Norm{\vy}_2^2}/(N\sigma^2)$ where $\sigma$ is the noise
variance.

Surprisingly, the least-square solution $\vX^{\#}$ seems also to be the smallest rank solution, i.e., numerically a
regularization with the trace norm of $\vX$, to promote a low-rank solution, does not yield to better results or even
lower rank solutions.  Although, the authors can not give an analytic stability result of the above algorithm, the
reconstruction from noisy observations gives reasonable stability, as can be seen in \figref{fig:simulation}.  Here, we
draw $\vxone$ and $\vxtwo$ from an i.i.d. Gaussian distribution with unit variance. If the magnitude of the first or
last coefficients is less than $0.1$ we dropped them from the simulation trial, this ensures full degree polynomials, as
demanded in the \thmref{thm:4correlation}.  As dimension grows, computation complexity increase dramatically and
stability decreases significant. Nevertheless, the MSE per dimension scales nearly linear with the noise power in dB.
Noticeable is the observation, that unequal dimension partitioning of $N$ yields to a better performance. 

\begin{figure}
  \centering
  \includegraphics[scale=0.54]{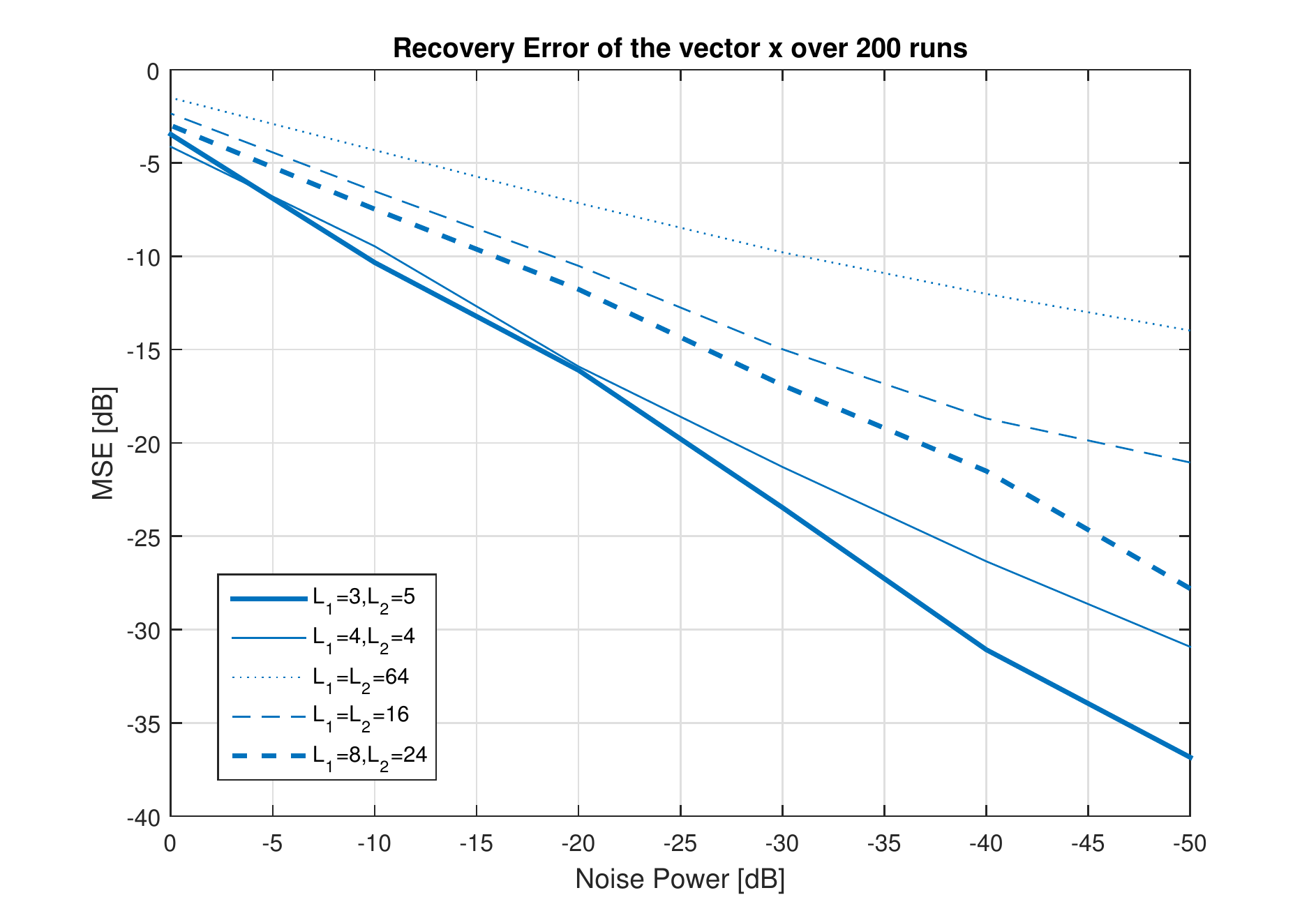}
  \caption{MSE/dim in dB for deconvolution of $\vxone\in\C^{L_1}$ and $\vxtwo\in\C^{\Ltwo}$ for various dimensions with
   additive Gaussian noise on the convolution and autocorrelations.}\label{fig:simulation}
\end{figure}

%%%%%%%%%%%%%%%%%%%%%%%%%%%%%%%%%%%%%%%%%%%%%%%%%%%%%%%%%%%%
\section{Conclusion} 
%%%%%%%%%%%%%%%%%%%%%%%%%%%%%%%%%%%%%%%%%%%%%%%%%%%%%%%%%%%%
We characterized the ambiguities of convolution by exploiting their polynomial factorizations. 
%For the autocorrelation we could omit all ambiguities if the factors are co-prime  
As an application we could derandomize a $4N-4$ auto and cross-correlation setup in \cite{JH16} by only
assuming a co-prime structure in $\vx$ and full degrees of the polynomials.
Moreover, we can provide a convex recovery algorithm which numerically also performs robust against additive noise.  

{\it Acknowledgments.} We would like to thank Kishore Jaganathan, Fariborz Salehi and Michael Sandbichler for helpful
discussions. A special thank goes to Richard Küng for discussing the dual certificate construction in dual problems in more
detail. This work was partially supported by the DFG grant JU 2795/3 and WA 3390/1. We also like to thank the Hausdorff
Institute of Mathematics for providing  for some of the authors resources at the Trimester program in spring 2016 on
``Mathematics of Signal Processing'' where part of the work have been prepared.

%%%%%%%%%%%%%%%%%%%%%%%%%%%%%%%%%%%%%%%%%%%%%%%%%%%%%%%%%%%%
\appendices
%%%%%%%%%%%%%%%%%%%%%%%%%%%%%%%%%%%%%%%%%%%%%%%%%%%%%%%%%%%%

\section{Proof of Lemma 3}\label{sec:prooflemma}

Usually, in the math literature, SDP problems are formulated for symmetric objects on symmetric cones over the real field
$\R$. This is due to the fact that minimizing or maximizing an objective function is only possible for real-valued
functions. Nevertheless, there is an extension to the complex case, which is sometimes called \emph{complex SDP
problems}. \seefor{ \cite[Sec.3]{GW04b}}
Let $\Alin:\C^{N\times N} \to \C^M$ be a linear map given by sensing matrices
$\{\vA_m\}_{m=0}^{M-1}\subset\C^{N\times N}$ (not necessarily Hermitian or symmetric). Moreover, we define the linear
objective function $\tr(\vC\vX)$ by a Hermitian
matrix $\vC\in H_N:=\set{\vA\in\C^{N\times N}}{\vA=\vA^*}$. Then the \emph{primal complex optimization problem} is given by
\begin{flalign}
  && 
  \min_{\vX\in\C^{N\times N},\vX\mgeq 0} \tr(\vC\vX) \quad \text{such that}\quad \Alin(\vX)=\vb &&
   \label{eq:primalC}
\end{flalign}
But $\Alin$ is not convex, since $\vb$ is not real-valued, if $\vA_m$ is not Hermitian. To obtain convex conditions,
we can just split imaginary and real part of $\vb$ by setting
\begin{align}
  \Alin_{R,m}(\vX)&=\frac{\Alin_m(\vX)+\cc{\Alin_m(\vX)}}{2}=\frac{\tr((\vA_m+\vA_m^*)\vX)}{2}
                    =\tr(\vA_{R,m}\vX)=\frac{b_m\!+\!\cc{b_m}}{2}=\Re(b_m),\\
  \Alin_{I,m}(\vX)&=\frac{\Alin_m(\vX)-\cc{\Alin_m(\vX)}}{2i}=\frac{\tr((\vA_m-\vA_m^*)\vX)}{2i}
                    =\tr(\vA_{I,m}\vX)=\frac{b_m\!-\!\cc{b_m}}{2i}=\Im(b_m),
\end{align}
for all $m\in [M]$. Hence, we yield $2M$ real-valued convex  measurements $\tAlin$ with the Hermitian sensing matrices $\vA_{I,m}$ and
$\vA_{R,m}$. This gives finally the equivalent \emph{primal complex convex optimization problem (primal complex SDP problem)}
\begin{flalign}
  &&\min_{\vX\mgeq 0} \tr(\vC\vX) \text{ such that}\quad \tAlin(\vX)=\begin{pmatrix}\Alin_{R}(\vX) \\
    \Alin_{I}(\vX)\end{pmatrix}= \begin{pmatrix}\Re(\vb)\\ \Im(\vb)\end{pmatrix}.&&\label{eq:primalcomplexSDP}
\end{flalign}
This complex SDP can be rewritten as a standard SDP over real-valued positive-semidefinite matrices in $S_N$, see for
example \cite[Sec.4]{GW04}. We therefor can assume the duality properties of the real SDP problems for the complex SDP
as well.
The \emph{dual convex optimization problem} is then given by
\begin{align}
  \max_{\vc,\vd\in\R^{M}} \skprod{\begin{pmatrix}\vc\\ \vd\end{pmatrix}}{\begin{pmatrix}\Re(\vb)\\\Im(\vb)\end{pmatrix}}
  \quad\text{s.t.}\quad\underbrace{\sum_{m=0}^{M-1}
  c_m \vA_{R,m} \!+\! d_m\vA_{I,m}}_{\in\range(\tAlin^*)}+ \vS = \vC, \vS\mgeq 0 \label{eq:dualC}.
\end{align}
\seefor{\cite[Sec.5]{Fre04b}}%
If $\vC=\zero$ the primal optimization problem \eqref{eq:primalC}  becomes a \emph{primal feasible
  problem\index{problem!primal feasible}} since any $\vX$ would yield the same objective value zero, which is equivalent
  to no objective function and hence to a \emph{primal complex feasible SDP problem}: 
\begin{flalign}
  &&  \text{find }\vX\mgeq 0 \quad \text{such that}\quad \tAlin(\vX)=\begin{pmatrix}\Re(\vb)\\ \Im(\vb)\end{pmatrix},
  &&\label{eq:feasiblepcp}
  \intertext{which is equivalent to}
  && \text{find }\vX\mgeq 0 \quad\text{such that}\quad \Alin(\vX)=\vb.  && \label{eq:feasiblepp}
\end{flalign}
Then the \emph{dual complex feasible problem\index{problem!dual feasible}} is given by, \cite[Sec.VI (12)]{JH16},
\begin{flalign}
  &&  \max_{\vlam\in\C^M}- \sum_{m} (\lam_m b_m + \cc{\lam_m}\cc{b_m} )\quad \text{such that}\quad \sum_{m} (\lam_m
  \vA_m + \cc{\lam_m}\vA_m^*)\mgeq 0, &&\label{eq:feasibledp}
\end{flalign}
which can be obtained by setting $-2\lam_m=c_m -id_m$ and $\vC=\zero$ in \eqref{eq:dualC}, since it holds 
\begin{align}
  -(\lam_m b_m+\cc{\lam_m b_m})&=\frac{c_m-id_m}{2}(\Re(b_m)+i\Im(b_m))+ \frac{c_m +id_m}{2}(\Re(b_m)-i\Im(b_m))\\
  &=c_m \Re(b_m) +d_m\Im(b_m)
\intertext{and}
\lam_m \vA_m +\cc{\lam_m}\vA_m^*&= (c_m\vA_m -id_m\vA_m +c_m\vA_m^* +id_m\vA_m^* )/2 =c_m \vA_{R,m} + d_m\vA_{I,m}.
\end{align}
The set of matrices 
\begin{align}
  \range(\tAlin^*)=\set{\vW=\sum_{m=0}^{M-1} (\lam_m \vA_m + \cc{\lam_m}\vA_m^*)}{\vlam\in\C^M}\label{eq:dualcertificate},
\end{align}
is the range space of $\tAlin^*$, which is indeed the
set of Hermitian matrices spanned by the Hermitian sensing matrices $\vA_{R,m}$ and $\vA_{I,m}$. Note, the real dimension is
less or equal to $2M$.

We will now  proof the central lemma for the uniqueness of the complex SDP program used in  \thmref{thm:4correlation}.
  
\begin{proof}[Proof of \lemref{lem:dualcertifrank2}]
  Note, that we have the equivalence
  \begin{align}
    \tAlin(\vX)=\begin{pmatrix}\Re(\vb)\\ \Im(\vb)\end{pmatrix} \quad\LRA \quad \Alin(\vX)=\vb
  \end{align}
  and therefor the range is equal, i.e., $\range(\Alin^*)=\range(\tAlin^*)$.
  One can insert the problem \eqref{eq:feasiblepp} directly  into Matlab with {\ttfamily cvx} toolbox, since it will be
  interpreted as the convex problem \eqref{eq:feasiblepcp} with real-valued constraints. We will use the version \eqref{eq:feasiblepp}
  since it is more natural for the proof.
  Let us assume $\vX^{\#}\mgeq 0$ is a feasible solution of the
  primal problem \eqref{eq:feasiblepp}, i.e., 
\begin{enumerate}[(a)]
  \item $\forall m\in [M]\colon \tr(\vA_m \vX^{\#})=b_m$ \label{enu:primalsol} 
\end{enumerate}
If we can show that $\vX^{\#}=\vx\vx^*$ is the only feasible solution then we have shown the unique solution.
Let us further assume  $\vlam\in\C^M$ is a solution of the dual complex feasible problem \eqref{eq:feasibledp}, i.e.,
\begin{enumerate}[(b)]
  \item $\vW=\sum_m \lam_m \vA_m + \cc{\lam_m}\vA_m^*\mgeq 0$\label{enu:dcpositive}
\end{enumerate}
then by the KKT conditions, strong duality see for example \cite[Thm.5.1]{Fre04b}, the solutions are the same if the
\emph{duality gap} is zero\footnote{This works with every $\vC\in H_N$ defining the  objective function $\tr(\vC\vX)$
(Note, that the dual certificate $\vW$ has to be also include $\vC$. For the feasible problem we have $\vC=\zero$},
i.e.
\begin{enumerate}[3.]
  \item $\tr(\vW\vX^{\#})=0\quad$ (Complementary slackness)\label{enu:slackness}
\end{enumerate}
By definition, $\vx\vx^*$ is a primal feasible solution. If we can construct a dual certificate $\vW$, which satisfy
\enuref{enu:dcpositive}, and which fulfills \enuref{enu:slackness}, then $\vx\vx^*$ is an optimal solution. Since
\eqref{eq:feasiblepcp} is a feasible problem every feasible solution is an optimal solution ($\vC=\zero$). But then for
every primal feasible solution $\vX^{\#}$ there must exists a dual certificate $\vW$ satisfying the slackness property
\enuref{enu:slackness}. We will use this condition to relax the  uniqueness condition. To ensure uniqueness of the
primal feasible solution $\vx\vx^*$ we have to show that no other primal feasible (optimal) solution $\vX^{\#}\mgeq0$
exist. This is equivalent to show that for any $\vX^{\#}=\vx\vx^*+\vH\mgeq 0$ given by any $\vH\in H_N$ it holds
\begin{align}
  & \Alin(\vX^{\#})=\vb \RA \vX^{\#}=\vx\vx^*,\label{eq:uniquecondorig}
\end{align}
which is by linearity of $\Alin$ equivalent to
\begin{align}
  \forall \vH\in H_N\text{ with }\vx\vx^*+\vH\mgeq 0 \text{ it holds: }\Alin(\vH)=\zero \RA \vH=\zero \label{eq:hzero}.
\end{align}
%
  \if0 If in addition a dual certificate $\vW$ exists which fulfills the slackness property, then $\vx\vx^*$ is the
  unique (optimal) solution.  \fi
  %
  \if0 In \cite[Sec.VI]{JH16} the  existence of a dual certificate $\vW\in\range(\Alin^*)$ was shown  which fulfills
  \enuref{enu:slackness} and \enuref{enu:dcpositive} for every primal solution $\vX^{\#}$. Moreover, they showed
  \enuref{eq:hzero}, which establish the uniqueness of the solution.  \newline \fi
%
%
To relax this to a more tractable condition we use an orthogonal decomposition of the set of Hermitian matrices $H_N$,
in an orthogonal sum, given by the tangent space at $\vx\vx^*$ to the manifold of Hermitian rank$-1$ matrices, defined
as
\begin{align}
  T_{\vx}:=\set{\vx\vh^* + \vh\vx^*}{\vh\in\C^N}
\end{align}
and its orthogonal complement $T^{\bot}_{\vx}$, i.e. $T_{\vx}\oplus T_{\vx}^{\bot}=H_N$ (note, $H_N$ is a real vector
space).  Let $\vX^{\#}\mgeq 0$ be a feasible primal solution, then we can write 
\begin{align}
  \vX^{\#}=\vx\vx^* +\vH =\vx\vx^* + \vH_{T_{\vx}} + \vH_{T_{\vx}^\bot}\mgeq0\label{eq:xrautedecomp}
\end{align}
for some $\vH\in H_N$. Then it holds
\begin{align}
  \quad\vH_{T_{\vx}^{\bot}}\ \bot\ T_{\vx} \LRA\quad &\forall\vh\in\C^N\colon\vH_{T_{\vx}^{\bot}}\ \bot\ \vx\vh^* +\vh\vx^*\\
  \LRA \quad &\forall \vh\in\C^N\colon \tr(\vH_{T_{\vx}^{\bot}} \vx\vh^*) +\tr(\vH_{T_{\vx}^{\bot}} \vh\vx^*)=0.
\intertext{Since $\vH_{T_{\vx}^{\bot}}$ is Hermitian this is equivalent to} 
 \LRA\quad &\forall \vh\in\C^N\colon 2\Re(\tr(\vH_{T_{\vx}^{\bot}} \vx\vh^*) )=0 \label{eq:vxzero}.
\end{align}
But this holds for all $\vh\in\C^N$ and hence also for $\vh=\vH_{T_{\vx}^{\bot}}\vx$ which implies
\begin{align}
  \vH_{T_{\vx}^{\bot}}\vx=\zero,\label{eq:htxvbotzero}
\end{align}
since $\tr(\vh\vh^*)=\Re(\tr(\vh\vh^*))\geq 0$.  It holds
\begin{align}
  \vH_{T_{\vx}^{\bot}} \mgeq 0 \LRA \forall \vz\in \C^N\colon \vz^*\vH_{T_{\vx}^{\bot}}\vz\geq 0. 
\end{align}
We can decompose $\vz\in\C^N$ for any $\vx\in\C^N$ in an orthogonal sum  $\spann(\vx)\oplus\spann(\vx)^\bot$ such that there
exists $\lam\in\C$ and $\vz_1\in\spann(\vx)^{\bot}$ with $\vz= \vx + \vz_1$.
Hence, 
\begin{align}
  \vH_{T_{\vx}^{\bot}} \mgeq 0
%   \forall \vz\in \C^N\colon \vz^*\vH_{T_{\vx}^{\bot}} \vz\geq 0
   \LRA &\forall \lam\in\C,\vz_1\in\spann(\vx)^\bot \colon (\lam\vx+
   \vz_1)^*\vH_{T_{\vx}^{\bot}}(\lam\vx+\vz_1)\geq 0,\label{eq:uniquenesscondition}
\intertext{which is by \eqref{eq:htxvbotzero} equivalent to}
\vH_{T_{\vx}^{\bot}} \mgeq 0\LRA
& \forall \vz_1 \in \spann(\vx)^{\bot} \colon \vz_1^* \vH_{T_{\vx}^{\bot}}\vz_1\geq 0.\label{eq:z1def}
\end{align}
But since we know that $\vX^{\#}\mgeq 0$ we get for all $\vz_1\in\spann(\vx)^{\bot}$ with \eqref{eq:xrautedecomp} 
\begin{align}
  0\leq \vz_1^*\vX^{\#}\vz_1=\vz_1^*\vx\vx^*\vz_1 + \vz_1^* (\vx\vh^* + \vh\vx^*)\vz_1 + \vz_1^* \vH_{T_{\vx}^{\bot}}\vz_1
  =\vz_1^* \vH_{T_{\vx}^{\bot}}\vz_1,
\end{align}
which proofs the positive-semi-definiteness of $\vH_{T_{\vx}^{\bot}}$ by \eqref{eq:z1def}.  Since $\vX^{\#}$ is a
feasible primal solution there must exists a dual certificate $\vW^{\#}\mgeq 0$ with $\tr(\vW^{\#}\vX^{\#})=0$. If we can
show that the dual certificate $\vW$ for $\vx\vx^*$ is the dual certificate for every feasible primal solution
$\vX^{\#}$, then the only feasible solution is $\vx\vx^*$ and we are done. To do so, take a feasible
$\vX^{\#}=\vx\vx^*+\vH$. Then  
\begin{align}
  \Alin(\vH)=\zero
\intertext{which is equivalent to}
  \forall m\in[M]: \tr(\vA_m\vH)=0
\end{align}
and also $\tr(\vA_m^*\vH)=0$ since $\vH=\vH^*$. Then we can take an arbitrary $\vlam\in\C^M$ which defines $\vW$ and get
\begin{align}
  0=\sum_{m}\lam_m\tr(\vA_m\vH) + \cc{\lam_m}\tr( \vA_m^*\vH)=\tr((\sum_m \lam_m \vA_m +
    \cc{\lam_m}\vA_m^*)\vH)=\tr(\vW\vH)=\tr(\vW\vH_{T_{\vx}})+\tr(\vW\vH_{T_{\vx}^{\bot}})\label{eq:whzero}.
\end{align}
By condition  \enuref{enu:zero1} in the Lemma we have $\vW\vx=\zero$ and hence it follows
\begin{align}
  \tr(\vW\vH_{T_{\vx}^{\bot}})=0.
\end{align}
But since $\vW\vx=\zero$ by condition \enuref{enu:zero1} and $\vH_{T_{\vx}^{\bot}}\vx=\zero$ by \eqref{eq:htxvbotzero}
both matrices share a one-dimensional subspace of their nullspaces. But  $\vW\mgeq 0$ with \enuref{enu:positive1} and $\vH_{T_{\vx}^{\bot}}\mgeq 0$ by \eqref{eq:uniquenesscondition} it follows
$\range(\vW)\subseteq \kernel(\vH_{T_{\vx}^{\bot}})$, which implies $\kernel(\vH_{T_{\vx}^{\bot}})=\C^N$ since
$\rank(\vW)=N-1$ by \enuref{enu:rank1} and therefor
$\vH_{T_{\vx}^{\bot}}=\zero$. This gives the three conditions in \enuref{enu:dualcertif} of the Lemma. Hence the
uniqueness condition \eqref{eq:uniquecondorig} relaxes  to condition \enuref{enu:nullset} as
\begin{align}
  \forall \vH_{T_{\vx}}\in T_{\vx} \colon \Alin(\vH_{T_{\vx}})=\zero \RA \vH_{T_{\vx}}=\zero.
\end{align}
Then $\vX^{\#}=\vx\vx^*$ is the unique solution of \eqref{eq:feasiblepp} and
hence \eqref{eq:primalC}.  
\end{proof}

\section{Sylvester Matrix}\label{app:sylvester}

The $N\times N$ Sylvester matrix  of two vectors $\va\in\C_0^{L_1+1},\vb\in\C_0^{L_2+1}$ with $N=L_1+L_2$ play the
crucial role in our analysis and are defined for $L_1\leq L_2$ as

\begin{align}
  \notag\\[-10pt]
  \vS_{\va,\vb}:=
  \left(\begin{array}{cccc|cccccr}
    \tzm{u1a}b_{0} & 0         & \dots &  0\tzm{u1b}& \tzm{u2a}a_{0}& 0        & \dots & 0 &\dots &0\tzm{u2b} \tzmr{r1a} \\
    b_{1}          & b_{0}     & \dots &  0          & a_{1}         & a_{0}    & \dots & 0  & \dots &0\\
    \vdots         &           & \ddots&       \vdots     & \vdots        & \vdots   & \ddots & \vdots & &\vdots\tzmr{r1b}\\
    \vdots         & \vdots    &       & \vdots     & a_{L_1}      &  a_{L_1-1}&\dots &  a_{0} & \dots & 0\tzmr{r2a}\\
    b_{L_2}        & b_{L_2-1} & \dots & b_{L_2-(L_1-1)}  & 0      & a_{L_1}   &  \dots& a_1 &\dots & 0 \\
    0              & b_{L_2}   & \dots & b_{L_2-(L_1-2)}  & \vdots  &  & \ddots   & & &\vdots\\
    \vdots         &           & \ddots&   \vdots      & \vdots        &\vdots    &  & &\ddots&\vdots\\
    0              & 0         & \dots &  b_{L_2}    &  0       & 0        & \dots & 0 & \dots& a_{L_1} \tzmr{r2b}
  \end{array}\right)\label{eq:sylvestermatrix}\\[-15pt]\notag
  \begin{tikzpicture}[overlay, remember picture,decoration={brace,amplitude=5pt}]
      % oberhalb der matrix
      \draw[decorate,thick] (u1a.north) -- (u1b.north)
      node [midway,above=5pt] {${\scriptsize{L_1}}$};
       \draw[decorate,thick] (u2a.north) -- (u2b.north)
       node [midway,above=5pt] {${\scriptsize{L_2}}$};
       \end{tikzpicture}
       \begin{tikzpicture}[overlay, remember picture,decoration={brace,amplitude=5pt}]
       \draw[decorate,thick] (r1a.east) -- (r1b.south east) 
       node [midway,right=5pt] {${\scriptsize{L_1}}$};
        \draw[decorate,thick] (r2a.east) -- (r2b.east)
        node [midway,right=5pt] {${\scriptsize{L_2}}$};
      % \draw[decorate,thick] (r1a.west) -- (r1b.west)
      % node [midway,above=0pt,rotate=270,xshift=0pt,yshift=5pt] {${\footnotesize{k-N+1}}$};
  \end{tikzpicture}
\end{align}
where the first $L_1$ columns are down-shifts of the vector $\vb$ and the last $L_2$ columns are down-shifts of the
vector $\va$, see for example \cite[Sec.VII]{JH16} or \cite[Def.7.2]{GCL92},\cite[(1.84)]{Bar83} (here they define the
transpose version and for polynomials $a(z):=\sum_{k}a_{L_1-k}z^k=z^{L_1}A(z)$ and
$b(z):=\sum_{k}b_{L_2-k}z^k=z^{L_2}B(z)$ with degree $L_1$ respectively $L_2$, which has no effect on the resultant
(determinant) or rank).
The \emph{resultant} of the polynomials $a$ and $b$ is the determinant of the Sylvester matrix
$\vS_{a,b}=\vS_{\va,\vb}$. \namen{Sylvester} showed that the two polynomials have a common factor (non-trivial, i.e.,
not a constant) if and only if $\det(\vS_{a,b})\not=0$, which is equivalent of having full rank, i.e.,
$\rank(\vS_{a,b})=N:=L_1+L_2$. This can be generalized to the degree of the greatest common factor (GCD), see
\cite[Thm.1.8]{Bar83}.
\begin{theorem} \label{thm:sylvesterrank}
  Let $a,b\in\C[z]$ with degree $L_1$ and $L_2$ generating the  Sylvester matrix $\vS_{a,b}$, then 
  the greatest common factor of $a,b$ has degree
  \begin{align}
    D=L_1 + L_2 - \rank(\vS_{a,b}).
  \end{align}
\end{theorem}
Multiplying the polynomials $\uX_1$ and $\uX_2$ in \thmref{thm:4correlation} by $z^{-1}$ is equivalent to adding a zero
to the coefficient vectors $\vx_1$ and $\vx_2$, hence we set
\begin{align}
  \va=-\vx_1^0:=\begin{pmatrix}-\vxone\\ 0 \end{pmatrix} \quad,\quad 
  \vb=\vx_2^0:=\begin{pmatrix}\vxtwo\\ 0 \end{pmatrix}.
\end{align}
Then the nullspace of $\vS:=\vS_{a,b}=\vS_{-\vx_1^0,\vx_2^0}$, which dimension is given by \thmref{thm:sylvesterrank} as
$D$, determines the set of convolution equivalences, since we have (see also \appref{app:dualcerti})
\begin{align}
  \vS\begin{pmatrix}\vtx_1\\ \vtx_2\end{pmatrix}= \begin{pmatrix}\vx_2*\vtx_1 - \vx_1*\vtx_2 \\ 0\end{pmatrix}
  =\zero_N\label{eq:sylvesternullspace}
\end{align}
with vectors $\vtxone\in\C^{L_1},\vtxtwo\in\C^{L_2}$. Hence, if  the polynomials $\uX_1$ and $\uX_2$ do not have a
common factor ($a$ and $b$ have common factor $z$ of degree $D=1$), then by \thmref{thm:sylvesterrank} the rank of $\vS$
is $N-1$, i.e., there exists only one pair $(\vtxone,\vtxtwo)\in\C^{\Lone}\times \C^{\Ltwo}$ up to a global scaling, for
which their convolutions are equal, i.e., 
\begin{align}
  \vxone*\vtxtwo=\vxtwo*\vtxone \quad\LRA \quad\vtxone=\lam\vxone,  \vtxtwo=\lam\vxtwo \ \text{ for some } \ \lam\in\C.
\end{align}
Usually this result is written in the polynomial or $z-$domain as 
\begin{align}
  \vxone *\vtxtwo = \vxtwo*\vtxone \quad\LRA\quad \uX_1 \tuX_2=\uX_2\tuX_1.\label{eq:sylvesterpoly}
\end{align}
where $\tuX_1$ and $\tuX_2$ are polynomials of degree $\leq L_1$ respectively $\leq L_2$. Hence, if $\uX_1$ and $\uX_2$
are co-prime the only possible polynomials are $\tuX_1=\lam\uX_1$ and $\tuX_2=\lam\uX_2$, up to a unit $\lam$ (trivial
polynomial), which becomes the scalar factor for $\vx$.  Hence the nullspace of $\vS$ is one-dimensional and therefor
$\rank(\vS)=N-1$.
\subsection{Dual Certificate Construction}\label{app:dualcerti}

To show that $\vS^*\vS$ is a dual certificate, we have to define $\vlam\in\C^{4N-4}$ such that by
\eqref{eq:dualcertificate} we get
\begin{align}
  \vS^*\vS = \sum_{m=0}^{4N-5} \lam_m \vA_m + \lam_m \vA_m^* = \sum_{i,j=1}^2
  \sum_{k=0}^{L_i+L_j-2}  \lam_{i,j,k}\vA_{i,j,k} +   \cc{\lam_{i,j,k}}\vA_{i,j,k}^T,
\end{align}
where we split again $\vlam^T=(\vlam_{1,1}^T, \vlam_{2,2}^T,\vlam_{1,2}^T, \vlam_{2,1}^T)$ in four blocks corresponding
to the $\vA_{i,j,k}$ in \eqref{eq:a11}-\eqref{eq:a21}. To derive the $\vlam_{i,j}$ we need to write $\vS^*\vS$ in block
structure.  Let us define the lower banded Toeplitz matrix generated by $\vx_i$ as
\begin{align}
  \vC_{\vx_i} &= \sum_{m=0}^{L_i-1} x_{i,m} \vT^m_{N-1},
\end{align} 
where $\vT^m_{N-1}$ is the $m$th  $(N-1)\times (N-1)$ shift-matrix (elementary Toeplitz matrices) defined in
\eqref{eq:timeshiftmatrix}.
To apply this on $\vx_j\in\C^{L_j}$ we have to embed $\vx_j$ in $N-1$ dimensions with the $(N-1)\times L_j$ embedding matrix as
defined in \eqref{eq:timeshiftmatrix} by
\begin{align}
  \vC_{\vx_i}^{_j} &= \sum_{m=0}^{L_i-1} x_{i,m} \vT^m_{N-1} \Proj_{N-1,L_j}.
\end{align} 
Here, the upper index $j$ refers to the embedding dimension $L_j$.  We then obtain the matrix notation for the linear
convolution \eqref{eq:convtime} between $\vx_i\in\C^{L_i}$ and $\vx_j\in\C^{L_{j}}$ as
\begin{align}
  \vC_{\vx_i}^{_j} \vx_j&= \vx_i*\vx_j.
%\intertext{and the conjugate-time-reversal of the linear correlation as}
%\vC_{\vx_i}^{_j *} \vx_j&= \cc{\vx_i^-}*\vx_j.
\end{align}
Hence, the Sylvester matrix $\vS=\vS_{-\vx_1^0,\vx_2^0}$ is the concatenation of the two lower banded  matrices
$\vC_{\vx_2^0}^{_1}$ and $\vC_{-\vx_1^0}^{_2}$ and we
get for any $\vtx_1\in\C^{L_1}$ and $\vtx_2\in\C^{L_2}$ the convolution products embedded in $N=L_1+L_2$ dimensions as
\begin{align}
  \vS_{-\vx_1^0,\vx_2^0} \begin{pmatrix}\vtx_1\\ \vtx_2\end{pmatrix} = \begin{pmatrix}\vC_{\vx_2^0}^{_1} & \vC_{-\vx_1^0}^{_2}
  \end{pmatrix}\begin{pmatrix}\vtx_1\\\vtx_2\end{pmatrix}= 
  \begin{pmatrix}\vC_{\vx_2}^{_1} & -\vC_{\vx_1}^{_2}\\ \zero^T_{L_1} & \zero^T_{L_2}\end{pmatrix}\begin{pmatrix}\vtx_1\\
    \vtx_2\end{pmatrix}  
      = \begin{pmatrix} \vx_2*\vtx_1 - \vx_1*\vtx_2 \\ 0 \end{pmatrix}.
      \label{eq:sylvesterbanded}
\end{align}
If we consider the product of the adjoint Sylvester matrix with itself we get %
\begin{align}
  \vS_{-\vx_1^0,\vx_2^0}^* \vS_{-\vx_1^0,\vx_2^0} = 
  \begin{pmatrix}   \vC_{\vx_2^0}^{_1 *} \\ \vC_{-\vx_1^0}^{_2 *} \end{pmatrix}
  \begin{pmatrix}   \vC_{\vx_2^0}^{_1} & \vC_{-\vx_1^0}^{_2} \end{pmatrix}
  =
  \begin{pmatrix}  
    \vC_{\vx_2^0}^{_1 *}\vC_{\vx_2^0}^{_1} & -\vC_{\vx_2^0}^{_1 *}\vC_{\vx_1^0}^{_2}\\
    -\vC_{\vx_1^0}^{_2 *}\vC_{\vx_2^0}^{_1} & \vC_{\vx_1^0}^{_2 *}\vC_{\vx_1^0}^{_2}\\
  \end{pmatrix}.
  \label{eq:sylvesterrangestar}
\end{align}
Since we have for $i,j\in\{1,2\}$ 
\begin{align}
  \vC_{\vx_i^0}^{_j*} = \sum_{m=0}^{L_i-1} \cc{x_{i,m}}  \Proj_{L_j,N}\vT^{-m}_N
\end{align}
we get for each of the four blocks in \eqref{eq:sylvesterrangestar} denoted by $i,i',j,j'$ as
\begin{align}
  \vC_{\vx_i^0}^{_j*} \vC_{\vx_{i'}^0}^{_{j'}}
   = \sum_{m=0}^{L_i-1} \sum_{l=0}^{L_{i'}-1}\cc{x_{i,m}} x_{i',l}\Proj_{L_j,N} \vT^{-m}_N  \vT^{l}_N\Proj_{N,L_{j'}} 
   = \sum_{m=0}^{L_i-1} \sum_{l=0}^{L_{i'}-1}\cc{x_{i,m}} x_{i',l}\Proj_{L_{j},N} \vT^{l-m}_N \Proj_{N,L_{j'}}.
\end{align}
Let us emphasize that $l,m$ are limited by $\pm L_1$ resp. $\pm L_2$, and since we consider the $L_1$ resp. $L_2$
embeddings, the zeros on the $l-m$th diagonal in $\vT^{l-m}_N$ can be ignored. By substituting $k=l-m$ we get
\begin{align}
  \vC_{\vx_i^0}^{_j*} \vC_{\vx_{i'}^0}^{_{j'}} 
    & =\sum_{k=-L_i+1}^{L_{i'}-1} \sum_{m=0}^{L_i-1} \cc{x_{i,m}} x_{i',k+m} \Proj_{L_j,N}\vT^{k}_N\Proj_{N,L_{j'}}
    =\sum_{k=0}^{L_i+L_{i'}-2} \! \! \!\underbrace{ \sum_{m=0}^{L_i-1} \cc{x_{i,m}}
    x_{i',k-L_i+1+m}}_{=\cc{(\vx_i*\cc{\vx_{i'}^-})_{L_i\!+\!L_{i'}\!-\!2\!-\!k}}=(\vx_{i'}*\cc{\vx_i^-})_k} 
    \!\! \Proj_{N,L_j}^T\vT^{k-L_i+1}_N\Proj_{N,L_{j'}}
\end{align}
where the inner sum is the correlation between $\vx_{i'}$ and $\vx_{i}$ at index $k\in [L_i+L_{i^{'}}-1]$.
Hence we get for the autocorrelations $i=i'\in\{1,2\}$ the diagonal blocks in \eqref{eq:sylvesterrangestar}
\begin{align}
  \vC_{\vx_2^0}^{_1*} \vC_{\vx_{2}^0}^{_{1}} &= \sum_{k=0}^{2L_2-2} (\vx_2*\cc{\vx_2^-})_k \Proj_{N,L_1}^T \vT_N^{k-L_2+1}
  \Proj_{N,L_1}= \sum_{k} (\va_{2,2})_{2L_2-2-k} \vT^{(k,2)}_{L_1,L_1}= \sum_k (\va^-_{2,2})_k \vT_{L_1,L_1}^{(k,2)},\\
  \vC_{\vx_1^0}^{_2*} \vC_{\vx_{1}^0}^{_{2}} &= \sum_{k=0}^{2L_1-2} (\vx_1*\cc{\vx_1^-})_k \Proj_{N,L_2}^T \vT_N^{k-L_1+1}
  \Proj_{N,L_2}= \sum_{k} (\va_{1,1})_{2L_1-2-k} \vT^{(k,1)}_{L_2,L_2}=\sum_k (\va_{1,1}^-)_k \vT_{L_2,L_2}^{(k,1)},
\end{align}
where $\vT_{L_i,L_i}^{(k,j)}$ was defined in \eqref{eq:LiLj}, but with the difference that $i\not=j$. Since the
dimensions of the $L_j \times L_{j}$ block matrices in \eqref{eq:sylvesterrangestar} are not fitting with the
autocorrelations $\vx_i * \cc{\vx_{i}^-}$ on the diagonal we have to cut respectively zero-pad the $\vlam_{i,i}$
correspondingly.  Let us assume w.l.o.g. that $L_1\leq L_2$. Then we set
\begin{align}
  \begin{split}
    \vlam_{1,1}&:=\frac{1}{2}\{(\va_{2,2}^-)_k\}_{k=L_2-L_1}^{L_2+L_1-2} \in\C^{2L_1-1}, \\
    \vlam_{2,2}&:=\frac{1}{2}\begin{pmatrix} \zero_{L_2-L_1}\\ \va_{1,1}^-\\ \zero_{2L_2-1- L_1}\end{pmatrix}
    \in\C^{2L_2-1},\\
\end{split}\label{eq:appvlam}
\end{align}
which gives by the conjugate-symmetry of the autocorrelations
\begin{align}
  \sum_{i=1}^2 \sum_{k=0}^{2L_i-2}\lam_{i,i,k}\vA_{i,i,k}  + \cc{\lam_{i,i,k}}\vA_{i,i,k}^T  =
  2\sum_{i=1}^2\sum_{k} \lam_{i,i,k}\vA_{i,i,k},\label{eq:lam11}
\end{align}
where the transpose of $\vA_{i,i,k}$ is equivalent to a time-reversal of $\vlam_{i,i}$, i.e.
$\lam_{i,i,k}^-=\lam_{i,i,2L_i-2-k}$.  For the anti-diagonal in \eqref{eq:sylvesterrangestar} we have
\begin{align}
  -\vC_{\vx_1^0}^{_2*} \vC_{\vx_{2}^0}^{_{1}} &= -\sum_{k=0}^{L_1+L_2-2} (\vx_2*\cc{\vx_1^-})_k \Proj_{N,L_2}^T
  \vT_N^{k-L_1+1}    \Proj_{N,L_1}= -\sum_{k} (\va_{2,1}^-)_k \vT^{(k)}_{L_2,L_1},\\
  -\vC_{\vx_2^0}^{_1*} \vC_{\vx_{1}^0}^{_{2}} &= -\sum_{k=0}^{L_1+L_2-2} (\vx_1*\cc{\vx_2^-})_k \Proj_{N,L_1}^T
  \vT_N^{k-L_2+1}
    \Proj_{N,L_2}= -\sum_{k} (\va_{1,2}^-)_k \vT^{(k)}_{L_1,L_2},
\end{align}
denoting the time-reversal of the cross-correlations. Hence we set similar
\begin{align}
  \vlam_{1,2}&:= -\frac{1}{2} \va_{2,1}^-\in\C^{N-1}\\
  \vlam_{2,1}&:= -\frac{1}{2}\va_{1,2}^-\in\C^{N-1}.
\end{align}
Since the off-diagonal matrices satisfy $\vA_{2,1,k}^*=\vA_{2,1,k}^T= \vA_{1,2,L_1+L_2-2-k}$ for $k\in[L_1+L_2-1]$ we
have again the transpose is equivalent to a time reversal of $\vlam_{2,1}$ and since $\cc{\vlam_{2,1}^-}= \vlam_{1,2}$
we have 
\begin{align}
  \begin{split}
  \sum_k \lam_{1,2,k}\vA_{1,2,k}  + \cc{\lam_{1,2,k}}\vA_{1,2,k}^T  &=-\frac{1}{2}\sum_k (\va_{2,1}^-)_k \vA_{1,2,k}
  +\cc{(\va_{2,1}^-)_k} \vA_{1,2,k}^T \\
  &= -\frac{1}{2}\sum_k (\va_{2,1}^-)_k\vA_{1,2,k}   +(\va_{1,2})_{k} \vA_{2,1,N-2-k}\\
  &= -\frac{1}{2}\sum_k (\va_{2,1}^-)_k\vA_{1,2,k}   +(\va_{1,2}^-)_{k} \vA_{2,1,k}
\end{split}\label{eq:lam12}\\
  \begin{split}
    \sum_k \lam_{2,1,k}\vA_{2,1,k}  + \cc{\lam_{2,1,k}}\vA_{2,1,k}^T  
    &=-\frac{1}{2}\sum_k (\va_{1,2}^-)_k \vA_{2,1,k} +\cc{(\va_{1,2}^-)_k}\vA_{2,1,k}^T   \\
    & =-\frac{1}{2}\sum_k (\va_{1,2}^-)_k\vA_{2,1,k}    +(\va_{2,1})_k\vA_{1,2,N-2-k}    \\
    & =-\frac{1}{2}\sum_k (\va_{1,2}^-)_k\vA_{2,1,k}    +(\va_{2,1}^-)_k\vA_{1,2,k}. 
  \end{split}\label{eq:lam21}
\end{align}
Hence, adding \eqref{eq:lam11},\eqref{eq:lam12} and \eqref{eq:lam21} yields $\vW=\vS^*\vS$.

%\bibliographystyle{IEEEbib}
%\bibliography{refs,jabref_philipp_utf2}
\newpage
\section*{References}
\printbibliography

\end{document}